\documentclass[acmsmall,screen,nonacm]{acmart}
\usepackage{tikz}
\usepackage{placeins}
\usepackage{array}
\setcopyright{none}
\newcommand{\agenra}{Agenra}
\newcommand{\Allow}{\mathsf{Allow}}
\newcommand{\Deny}{\mathsf{Deny}}
\newcommand{\Pause}{\mathsf{Pause}}
\newcommand{\dom}{\operatorname{dom}}
\newcommand{\kw}[1]{\textsf{#1}}
\title[Agenra: Stateful LLM Agent Authorization]{Agenra: A Language and Verified Core for\texorpdfstring{\\}{ }Stateful LLM Agent Authorization}
\author{Basavesh Ammanaghatta Shivakumar}
\affiliation{\institution{Washington State University}\city{Pullman}\country{USA}}
\author{Swarn Priya}
\affiliation{\institution{Washington State University}\city{Pullman}\country{USA}}
\author{Peng Gao}
\affiliation{\institution{Virginia Tech}\city{Blacksburg}\country{USA}}
\keywords{policy revision, consumable authority, agent security, bounded verification}
\ccsdesc[500]{Security and privacy~Information flow control}
\ccsdesc[300]{Software and its engineering~Formal software verification}
\begin{document}
\begin{abstract}
Applications that deploy LLM agents keep revising the permissions their agents
use. When authority is consumable, even widening a permission can break a
legitimate workflow, because an early call may use up a grant that a later call
needs. Checking individual requests, or watching one run succeed, misses this.
We present Agenra, a policy language with a Lean enforcement core that tracks the
values each call uses, resources discovered at run time, and authority consumed
by earlier calls. Agenra treats safety and workflow preservation as separate
properties.

Our key idea is to certify a revision with a simulation between the old and new
permission ledgers. Quota increases, appended grants and independent reorderings
admit such simulations by construction, and Lean checks proposed simulations for
other finite-quota revisions. The proofs use the deployed matching predicates and
extend through the full monitor, covering discovered resources, payloads that are
never enumerated, failed calls and approvals. They assume serialized calls and
identical service results under both policies, and they trust the Python
translation layer.

On 51 authored revisions, the checks certify all 21 that preserve their clients
and none of the 30 that do not, and every lost workflow is confirmed by replay.
They also certify a real quota increase from an open-source agent application.
In a report-delivery application, raising the send quota preserves every client
yet permits a second approved send, so preservation does not imply safety.
\end{abstract}
\maketitle
\section{Introduction}
\label{sec:introduction}
Applications must maintain the permissions under which agents use files,
accounts and services, and they revise them routinely: one open-source AI app
builder doubled its agent's free-message quota in a single pull
request~\cite{dyadquota2026}. Because agents choose calls at run time, a revision
must keep every declared workflow working, including ones never exercised in
testing. A permissive change can nevertheless break a legitimate workflow when an early call consumes authority needed
later. We study how to preserve these workflows while checking that revisions
admit no forbidden effects.

Consider a task that may
delete A and B once each. Permissions are searched in order: first one deletion
of B, then one of A. Deleting A followed by B succeeds. Extend the first
permission to cover either object. Under first-fit allocation, A now uses it,
leaving no permission for B. Yet A can be deleted twice, through different
permissions. Adding a resource can therefore block a legitimate workflow and
admit a forbidden one.
Barth and Mitchell observed such losses for digital-rights
licenses~\cite{barthdrm2006}; we study how to check revisions that avoid them in
agent workflows.

In agent workflows, usage limits interact with data flow and discovered
information. Suppose an agent must project the public fields of a private
report, look up Alice's address and send the result once. Untrusted
content may ask it to send the raw report, substitute another address or
retry~\cite{greshake23}. The application must track the body's origin, the
designated lookup and earlier attempts. Information-flow defenses already
control recipients through capabilities and reader sets
~\cite{denning76,myersliskov,camel,fides,pcas}. A Public classification alone,
however, expresses neither this task's designated recipient nor its remaining
attempts. We study how these checks evolve together across calls, and whether
revising their permissions preserves complete workflows.

\emph{Safety} asks whether a policy excludes forbidden effects.
\emph{Client preservation} asks whether a revision retains every admitted call
sequence under the same external observations. Neither property implies the
other. The widening above loses a client and admits a forbidden trace, while
appending a second A grant keeps every client but also admits AA. We establish
preservation with a simulation between the old and new ledgers: if every
successful charge under the old permissions can be matched by a successful
charge under the new ones, with their usage counters staying related, then every
old client survives at every finite length. Quota increases, appended grants and
reorderings of independent grants admit such relations by construction, so we
check them structurally; for other finite-quota revisions, Lean checks a
proposed relation for closure.

What makes this harder than for static licenses is the state the simulation
must track. Matching depends on resource bindings discovered at run time, so a
certificate must hold in every binding state. Arguments are arbitrary JSON, so
it must cover payloads it never enumerates. The monitor also keeps flow, origin,
history and approval state, which must stay equal under the old and new policies
through failed calls, scans and approvals. We prove each of these steps in Lean
against the deployed matcher and monitor.

Agenra lets applications declare effects, usage limits, trusted lookups, and
rules for data flow and argument origins; each effect needs authorization from
both the user's permissions and the application ceiling, and a compiled Lean
core checks every call before Python invokes the tool. When no certificate
applies, bounded analysis searches the supplied finite model for forbidden
executions and useful completions, which we then replay on the tools.

The paper contributes:
\begin{enumerate}
\item Checked conditions under which a revision of ordered permissions preserves
every declared client, proved against the deployed matching predicates
(Sections~\ref{sec:contracts} and~\ref{sec:client-theorem}).
\item A policy language and deployed Lean core for consumable authority, discovered
resources and argument origins, sharing executable transitions with proved bounded
workflow analysis
(Sections~\ref{sec:design}, \ref{sec:semantics} and~\ref{sec:analysis}).
\item Experiments on policy maintenance, concrete tools and bounded agent
workflows that evaluate safety and useful completion separately: the checks
certify all 21 preserving revisions in an authored corpus of 51 and none of its
30 regressions, and raising the report's send quota preserves every client yet
permits a second approved send (Section~\ref{sec:eval}).
\end{enumerate}

\section{System and Threat Model}
\label{sec:threat}
The application installs rules for an authenticated user's task.
An agent proposes tool calls, and a reference monitor checks them before execution.
Tools can read, transform or send data. An attacker can place instructions in
documents, messages or tool results that influence the agent's operation,
arguments and call order. The goal is to block calls violating the installed
policy even when the agent deliberately proposes them.

\paragraph{What the application must supply correctly.}
The application defines tool schemas and effects, source classifications,
trusted lookup rules and approval issuance. Tools must perform only declared
effects and use the dependencies provided by their adapters.
The agent may call the execution API but may not change policy, register values
or issue approvals. The monitor and dispatch path must be protected from
modification. We enforce this API boundary with process-local Python object
identity, which does not isolate arbitrary Python code.

By default, calls are checked by a compiled Lean core whose functions carry proofs.
Python policy translation, opaque-handle mapping, JSON transport, dispatch
integration and the Lean compiler/runtime remain trusted.
Scanners and trusted transformations must do what they claim.
Bounded analysis also assumes that its models of tool results and reviewer
actions include every execution we draw conclusions about.

A trusted retrieval tool can still return attacker-written text.
A designated directory field has a different role: the application relies on
it for an authoritative address. A compromised directory can return the wrong
address while the monitor correctly follows its installed lookup rule.

\paragraph{Expressing the intended task.}
The application validates the structured interpretation of the user request.
Its digest supports auditing; it neither authenticates the interpretation nor
proves it correct.
Labels come from registered sources and transfer rules. Automatic intent
inference and labeling of arbitrary content are outside this paper.
If an intended requirement is missing from the policy, an independent task
oracle may reveal the mismatch, but the monitor cannot enforce the missing rule.

\paragraph{Execution assumptions.}
All relevant sends, writes, delegations and protected responses pass through
a modeled sink. The formal results assume one call at a time.
Authorization may be followed by a tool failure, so an admitted attempt spends
quota even if its outcome is unknown; nothing is rolled back.

\subsection{Mediation Architecture}
\label{sec:architecture}
\begin{figure}[htbp]
\centering
\begin{tikzpicture}[
  x=1cm,y=1cm,>=stealth,
  box/.style={draw=black!60,rounded corners=2pt,align=center,
    font=\footnotesize,inner sep=4pt,minimum height=1.05cm,fill=white},
  core/.style={box,draw=blue!55!black,fill=blue!5},
  theorem/.style={box,font=\scriptsize,minimum width=3.6cm,minimum height=1.3cm},
  flow/.style={->,semithick},
  authority/.style={->,semithick,dashed},
  proof/.style={densely dotted,thick,draw=blue!55!black},
  note/.style={font=\scriptsize,align=center,fill=white,inner sep=2pt}]

\node[box,minimum width=12cm,minimum height=.75cm] (app) at (6,0)
  {\textbf{Trusted application}: authenticated task, policies, tool schemas\\
   designated lookups, quota grants and approval issuer};

\fill[blue!4,rounded corners=4pt] (5,-.8) rectangle (8.6,-3.58);
\draw[blue!45!black,rounded corners=4pt] (5,-.8) rectangle (8.6,-3.58);
\node[box,minimum width=1.4cm] (agent) at (.95,-1.65)
  {\textbf{Agent}\\proposes\\tool calls};
\node[box,minimum width=1.95cm] (bridge) at (3.15,-1.65)
  {\textbf{Python bridge}\\lower policies\\forward calls};
\node[core,minimum width=3.25cm,minimum height=1.35cm] (runtime) at (6.8,-1.65)
  {\textbf{Deployed Lean core}\\resolve stored values\\
   check policy + contracts\\commit quotas + permit};
\node[box,minimum width=2.5cm] (tools) at (10.65,-1.65)
  {\textbf{External tools}\\Python invokes effects\\behavior trusted};
\node[core,minimum width=3.25cm,minimum height=.65cm] (finish) at (6.8,-3.05)
  {\textbf{Lean completion}\\labels, slots, handles};

\draw[authority] (app.south -| bridge.north) -- (bridge.north);
\draw[authority] (app.south -| runtime.north) -- (runtime.north);
\draw[flow] (agent) -- node[note,above] {call} (bridge);
\draw[flow] (bridge) -- node[note,above] {check} (runtime);
\draw[flow] (runtime) -- node[note,above] {Allow} (tools);
\node[note] at (3.15,-2.52) {Deny/Pause: no dispatch};
\draw[flow] (tools.south) -- (10.65,-3.05) --
  node[note,above] {result + scans} (finish.east);
\draw[flow] (finish.west) -- node[note,above] {successful result + handles}
  (.95,-3.05) -- (agent.south);

\draw[proof] (6.8,-3.58) -- (6.8,-3.95);
\fill[green!4,rounded corners=4pt] (-.05,-3.95) rectangle (12.05,-6);
\draw[green!35!black,rounded corners=4pt] (-.05,-3.95) rectangle (12.05,-6);
\node[font=\footnotesize\bfseries] at (6,-4.22)
  {Formal guarantees checked in Lean};
\node[theorem] at (1.95,-5.12)
  {\textbf{Deployed runtime invariants}\\
   ordered quotas; binding origin\\
   argument origins; one-use approvals\\
   atomic quota/approval commit};
\node[theorem] at (6,-5.12)
  {\textbf{Policy correspondence}\\
   flow admission $\leftrightarrow$ rule meaning\\
   history matching $\leftrightarrow$ intervals\\
   sequence, repeat, choice, negation};
\node[theorem] at (10.05,-5.12)
  {\textbf{Runtime traces and search}\\
   checked revision preservation\\
   confidential ancestry (scoped)\\
   bounded search $\leftrightarrow$ reachability};

\node[box,minimum width=5.85cm,minimum height=.95cm] at (2.95,-6.75)
  {\textbf{Lean workflow analysis}\\
   same transitions + proved bounded search\\
   finite supplied outcome/reviewer model};
\node[box,minimum width=5.85cm,minimum height=.95cm] at (9.05,-6.75)
  {\textbf{Python/Z3 analysis}\\
   finite-instance certificates + replay\\
   trusted export and runtime tables};
\node[note] at (6,-7.8)
  {\textbf{Trusted}: Python lowering/transport, source and dependency declarations,\\
   scanner judgments, external tool behavior and outcome-model coverage.};
\end{tikzpicture}
\caption{Runtime enforcement and formal guarantees. Solid arrows carry calls and
results; dashed arrows install authority. Blue identifies the deployed Lean core;
the dotted link identifies functions whose properties are proved beforehand.
The proof panel groups deployed and conditional guarantees
(Table~\ref{tab:proof-status}). Offline analysis does not authorize live calls.}
\label{fig:architecture}
\Description{An application installs policies and authority. An agent sends a
call through Python to the Lean core, which checks the call and commits quotas
and approvals before Python invokes a tool. Results and scans pass through Lean
completion before handles return. A proof panel lists deployed runtime invariants,
flow and nested history correspondence, client preservation, confidential ancestry
and bounded search correctness. Separate offline Lean and Python Z3 analyzers
have explicit model and certificate boundaries. Python translation, transport,
classifications, scanners and external tools remain trusted.}
\end{figure}

The application validates the contract against the authenticated request and
installs rules that tool responses cannot rewrite (Figure~\ref{fig:architecture}).
Python forwards each call to Lean, which checks it against the task state and
spends quota and approval before the tool runs, and which processes each result
before Python returns a handle. The runtime executes functions whose properties
were proved during development, while offline analysis explores supplied results.
Sanitizers and reviewers are trusted: the monitor checks that an approval exists
and is consumed, but not whether the reviewer judged correctly
(Table~\ref{tab:proof-status}).
Section~\ref{sec:motivation} follows this path for the report task.

\FloatBarrier
\section{How Agenra Handles the Report Task}
\label{sec:motivation}
We return to the introduction's task: \emph{send the public report fields to
Alice, at most once}. Before the agent starts, the application requires a
Public body produced by \texttt{Sanitize}, Alice's address from a designated
directory lookup, and at most one send attempt. Both the user
permissions and the application ceiling impose this limit.

\paragraph{The agent proposes a workflow.}
Each variable below names an immutable stored result, called a \emph{value
handle}. The agent uses handles in later calls; the adapter recovers their
values and tracks their sources.
\begin{quote}
\begin{minipage}{\linewidth}
\small\ttfamily
r := Read()\hfill // confidential report\par
s := Sanitize(record=r)\hfill // public field projection\par
d := Lookup()\hfill // designated lookup for Alice\par
a := Extract(record=d)\hfill // address\par
Send(body=s, recipient=a)\hfill // one send attempt
\end{minipage}
\end{quote}
\texttt{Sanitize} is application code that selects approved fields. Its output
is Public; the original remains Confidential. \texttt{Lookup} returns Alice's
record, and \texttt{Extract} obtains its address.

The installed contract names a slot \texttt{target}, filled only from the
successful designated \texttt{Lookup()} result's \texttt{destination} field.
Both permission lists authorize \texttt{message.send(@target)} once.
Increasing the lookup quota can support address revalidation; increasing the
send quota instead changes the task's at-most-once restriction. Section~\ref{sec:eval}
checks these maintenance changes on the same application.

\paragraph{Lean checks the call against the stored task state.}
Python reports successful tool results to Lean, which stores their values,
labels and dependencies. Suppose Alice's address is
\texttt{alice@example.org}, and the lookup binds \texttt{target} to it.
For \texttt{Send(body=s, recipient=a)}, Python maps
the handles to Lean identifiers and forwards the call. Lean retrieves the stored
values and checks their combined label, the body's \texttt{Sanitize} origin,
the actual recipient and remaining quota on both permission lists.
Calls that fail a check do not reach the tool.

If the checks pass, Lean charges quota before Python invokes the mail tool.
Failure still consumes quota. A conflicting or malformed directory result
invalidates its binding; an equal result preserves it.

\begin{table}[htbp]
\caption{What the report policy allows and rejects. The first four rows use
flow, argument-origin and contract checks. The last three add exact-call approval
(Section~\ref{sec:permit-encoding}), outside the base correspondence of
Theorem~\ref{thm:correspondence}; origin is also outside that fragment.
Rejected calls never reach the mail tool.}
\label{tab:motivation}
\centering\small
\begin{tabular}{p{0.44\linewidth}p{0.28\linewidth}p{0.18\linewidth}}
\toprule
Candidate & Obligation & Result \\
\midrule
Public summary to bound Alice address & Flow and contract & Allow \\
Raw report to Alice & Confidentiality & Reject \\
Second send, including retry after failure & Consumed quota & Reject \\
Public body from another producer to bound Alice address & Argument origin & Reject \\
\midrule
Edited summary with old approval & Exact value identity & Reject \\
Recomputed identical summary with old approval & Exact value identity & Reject \\
Second send with consumed approval & One-use permit & Reject \\
\bottomrule
\end{tabular}
\end{table}

\paragraph{The rules must express the intended task.}
Omitting the origin rule allows unrelated Public content to Alice: classification
alone does not require a \texttt{Sanitize} result. The evaluation retains this
original policy as a control. The origin rule fixes only the body's immediate
producer, so \texttt{Sanitize} applied to attacker content from \texttt{Other}
still passes and spends the single \texttt{record.release} grant. Adding
\texttt{origin report\_input \{via: Sanitize argument: record producer: Read\}}
rejects this decoy on both backends and keeps all six legitimate delivery orders;
reported results use Figure~\ref{fig:delivery-policy} without it. Enforcement
cannot infer such requirements from intent.

\paragraph{Approval and analysis build on the same workflow.}
An optional permit approves the exact call and is consumed before invocation
(Section~\ref{sec:exact-release}); Figure~\ref{fig:delivery-policy} gives the
policy. Offline analysis searches for forbidden sends and for useful deliveries,
since rejecting every send would also fail the task.

\section{The Policy Language}
\label{sec:design}
Agenra lets an application write rules for an agent before the task begins.
For report delivery, the rules describe where the private report comes from,
which output may be sent, and which address and number of sends are permitted.
The application supplies these rules; the agent chooses calls to propose.
Table~\ref{tab:lowering} connects the declarations to the records the monitor
maintains. We explain the language's building blocks below, then collect their
syntax in Figure~\ref{fig:syntax}.

\begin{table}[htbp]
\caption{Reading the report policy. Each declaration tells the monitor what to
record or check; the application must justify the classifications and tool behavior.}
\label{tab:lowering}
\centering\small
\begin{tabular}{p{0.27\linewidth}>{\raggedright\arraybackslash}p{0.28\linewidth}p{0.34\linewidth}}
\toprule
Application requirement & Declaration & Monitor meaning \\
\midrule
The projection produces public fields &
\texttt{Sanitize} edge: \texttt{transform \{set: \{sensitivity: Public\}\}} &
Mark the new result Public; keep the original report Confidential \\
Alice's address comes from the designated directory &
\texttt{resolve target: Lookup() -> destination} &
Store that result's address; invalidate it on a conflict or malformed result \\
The body must come from the projection &
\texttt{origin report\_body \{via: Send argument: body producer: Sanitize\}} &
Check the named argument's stored producer, independently of equal payloads \\
Send once to that address &
\texttt{message.send(@target)}, quota one on each side &
Compare the actual recipient with the stored address and check both usage limits \\
\bottomrule
\end{tabular}
\end{table}

\subsection{Naming Locations and Operations}
A \emph{node} names a place where data enters, is processed, or leaves the task.
These are respectively called a source, internal memory, and a sink.
The private report is a source; working memory holds intermediate results;
the mail endpoint is a sink. Each node has a declared default classification.

An \emph{edge} names an operation between such locations. For example,
\texttt{Read} takes the report into memory, \texttt{Sanitize} produces another
value in memory, and \texttt{Send} reaches the mail endpoint. Edge declarations
specify permitted endpoints and how the output is classified. An agent
declaration lists the operations that agent may use and its sensitivity ceiling.

\subsection{Classifying Values}
A \emph{label} is a record attached to a value. Its sensitivity says whether the
value is Public, Confidential, or another declared level. Its categories record
concerns such as credentials. Its trust flag records whether the origin is
declared trusted; it does not prove that the content is correct.
We write this record as $\ell=(s,C,t)$.

When an operation uses several tracked values, their classifications are
combined. Adding public text to a private report does not make the report
public. The combined sensitivity is the highest input sensitivity, categories
are accumulated, and an untrusted input makes the combination untrusted.
Formally, this combination, called the \emph{join}, is
\[
(s_1,C_1,t_1)\sqcup(s_2,C_2,t_2)
=(\max(s_1,s_2),C_1\cup C_2,t_1\land t_2).
\]
Category matching includes descendants (Appendix~\ref{app:labels}).

\subsection{Rules for Outputs and Earlier Calls}
\label{sec:propagation}
An edge's \emph{propagation rule} specifies the output label.
Inheritance keeps the combined input label, as for \texttt{Extract}.
A transform changes selected components. \texttt{Sanitize} assigns Public to
its result because the application supplies a projection that selects approved
fields. This assumption concerns the tool's actual behavior; changing a label
does not remove secret content. The original report remains Confidential.

A transform removes configured categories, subtracts a sensitivity decrement,
then applies explicit sensitivity and trust assignments, in that order.
Other propagation modes and scans have additional application assumptions
(Appendix~\ref{app:transfers}).

A \emph{flow rule} decides whether an operation may carry its output label to a
destination. The report policy allows only Public data to the mail sink.
No matching allow means denial, and a matching deny takes precedence.

A \emph{history rule} can check for \texttt{Review} between a matched
\texttt{Read} and later \texttt{Send}. It may examine the operations that
produced the current inputs, called their lineage, or all earlier admitted
calls in the session. Both include the proposed call. A matching review
interval does not establish that every input was reviewed; the label join
still detects a Confidential input in a merge.

\subsection{Checking the Values Actually Used}
For \texttt{Send(body=s, recipient=a)}, the monitor must check the values
named by those handles. The precise adapter retrieves immutable stored values
and derives both concrete arguments and their producing-call IDs. Thus the
agent cannot submit the private report while claiming to use the Public summary.
Successful results receive fresh identities, even when their contents match.
Hidden reads still require complete dependency declarations.

An \emph{origin rule} constrains one argument's producer. For example,
\texttt{origin report\_body \{via: Send argument: body producer: Sanitize\}}
requires the body handle to name a successful \texttt{Sanitize} result.
An optional \texttt{preserving: \{Copy\}} clause also accepts that origin through
explicitly declared operations with one distinct parent. A sanitized recipient
cannot justify an unrelated body; identical bytes from another producer do not
qualify. Section~\ref{sec:argument-origin} gives the inductive meaning and
its correspondence to the deployed check.

\paragraph{Syntax and formal meaning.}
In Figure~\ref{fig:syntax}, $n$ names a node, $\pi$ a node pattern, $\bar\pi$ a
result-field path, $x$ an argument, $v$ a finite JSON scalar, $\alpha$ an action,
and $k$ a nonnegative integer. Section~\ref{sec:semantics} defines decisions and state updates;
Section~\ref{sec:analysis} defines the symbolic subset $\mathcal F$ and separate
Lean workflow coverage.

\begin{figure}[htbp]
\centering\footnotesize
\setlength{\tabcolsep}{2pt}
\begin{tabular}{@{}r@{\ }c@{\ }l@{}}
\multicolumn{3}{@{}l}{\textbf{Schema and flow policy}}\\
$s$ & $\in$ & $\{\mathit{Public}<\mathit{Internal}<\mathit{Confidential}<\mathit{Restricted}<\mathit{TopSecret}\}$\hfill\textit{sensitivity}\\
$\ell$ & $::=$ & $(s,\ \{c_1,\ldots,c_n\},\ t)$\hfill\textit{label: sensitivity, category paths, trust bit}\\
$\lambda$ & $::=$ & $[\kw{sensitivity}:s]\ [\kw{category}:\{\overline{c}\}]\ [\kw{trusted}:t]\ [\overline{p{:}v}]$\hfill\textit{label fields, properties $p$}\\
$\mathit{node}$ & $::=$ & $\kw{node}\ T\!::\!\texttt{"}i\texttt{"} : \kappa\ \{\ \lambda\ \}$\quad $\kappa\in\{\kw{source},\kw{internal},\kw{sink}\}$\hfill\textit{endpoint with static label}\\
$\rho$ & $::=$ & $\kw{inherit}(\kw{from}) \mid \kw{transform}\{\overline{r}\}$\hfill\textit{label transfer (Appendix~\ref{app:transfers} adds opaque and constant)}\\
$r$ & $::=$ & $\kw{remove}:\{\overline{c}\}\mid\kw{downgrade}:\{\kw{sensitivity}:k\}\mid\kw{set}:\{\beta\}$\hfill\textit{transform clauses}\\
$\beta$ & $::=$ & $\kw{sensitivity}:s\mid\kw{trusted}:t\mid\kw{sensitivity}:s,\ \kw{trusted}:t$\hfill\textit{assignments, fixed evaluation order}\\
$\mathit{edge}$ & $::=$ & $\kw{edge}\ e\ \{\ \kw{from}: \epsilon,\ \kw{to}: \epsilon,\ \kw{propagation}: \rho\ \}$\quad $\epsilon ::= \kappa \mid \kappa(T)$\hfill\textit{registered operation}\\
$\mathit{agent}$ & $::=$ & $\kw{agent}\ a\ \{\ \kw{max\_sensitivity}: s\ \ \kw{edges}: [\overline{e}]\ \}$\hfill\textit{principal and ceiling}\\
$\mathit{flow}$ & $::=$ & $\kw{flow}\ (\kw{allow}\mid\kw{deny})\ f\ \{\ [\kw{from}:\pi]\ [\kw{to}:\pi]\ [\kw{via}: e]\ [\kw{carrying}: \{\phi\}]\ \}$\hfill\textit{clause over the candidate}\\
$\phi$ & $::=$ & $\kw{sensitivity}: \bowtie s \mid \kw{category}: c \mid \phi, \phi$\hfill\textit{carrying predicate, $\bowtie\in\{==,!=,\leq,\geq,<,>\}$}\\[3pt]
$\mathit{origin}$ & $::=$ & $\kw{origin}\ o\ \{\ \kw{via}:e\ \kw{argument}:x\ \kw{producer}:e'\ [\kw{preserving}:\{\overline{e}\}]\ \}$\hfill\textit{argument-specific producer requirement}\\[3pt]
\multicolumn{3}{@{}l}{\textbf{Tool effects and task contracts}}\\
$\mathit{tool}$ & $::=$ & $\kw{tool}\ e\ \{\ \kw{from}: n,\ \kw{to}: n,\ [\kw{arguments}: [\overline{x}]]\ [\kw{required}: [\overline{x}]]\ \kw{effects}: [\overline{\eta}]\ [\kw{authoritative}: [\overline{\bar\pi}]]\ \}$\\
$\eta$ & $::=$ & $\alpha(x) \mid \alpha(x{*}) \mid \alpha(\texttt{"}\mathit{res}\texttt{"})$\hfill\textit{resource from an argument, a list, or a constant}\\
$\mathit{contract}$ & $::=$ & $\kw{contract}\ C\ \{\ \overline{\mathit{res}}\ \ \ \kw{user}: [\overline{\mathit{perm}}]\ \ \ \kw{application}: (\kw{same} \mid [\overline{\mathit{perm}}])\ \}$\hfill\textit{two-sided authority}\\
$\mathit{res}$ & $::=$ & $\kw{resolve}\ \mathit{slot}: e(\overline{x{:}v}) \to \bar\pi$\hfill\textit{designated authoritative lookup}\\
$\mathit{perm}$ & $::=$ & $\alpha(\overline{\varrho})\ [\{\ \overline{o}\ \}]$\hfill\textit{effect permission; quota defaults to one}\\
$o$ & $::=$ & $\kw{args}: [\overline{x}]\mid\kw{quota}: (k\mid\kw{unlimited})\mid\kw{where}: \{\overline{\chi}\}$\hfill\textit{optional permission properties}\\
$\varrho$ & $::=$ & $\texttt{"}\mathit{res}\texttt{"} \mid @\mathit{slot}$\qquad
$\chi ::= x: \kw{in}(\overline{v}) \mid x: \kw{range}(v_1,v_2)$\hfill\textit{resource; argument constraint}\\
\end{tabular}
\caption{Selected concrete syntax of \agenra{}. Brackets around clauses denote optional syntax; those after a colon denote lists. Overbars denote comma-separated sequences; declaration-property commas are optional. $\ell$ is a semantic label and $\lambda$ its declaration fields. Node patterns $\pi$ may also carry \texttt{where} tests (not shown) on labels and properties, such as \texttt{sink where \{external == true\}}; \texttt{carrying} tests only sensitivity and categories. Origin, history and opaque/constant transfers are outside $\mathcal F$ (Appendix~\ref{app:formal-coverage}); the Lean workflow profile supports additional constructs. Exact-call approval is installed at run time (Section~\ref{sec:exact-release}).}
\label{fig:syntax}
\Description{Grammar of node, edge, agent, flow, argument-origin, tool-effect and task-contract declarations. Tool schemas connect concrete arguments to effects; contracts contain ordered permissions and designated resource resolvers.}
\end{figure}

\section{Task Contracts: Permitted Actions and Limits}
\label{sec:contracts}
A Public summary can still be sent to the wrong person. Data classification
therefore answers only part of the authorization question. A \emph{task contract}
also states which actions may affect which resources, with which arguments and
how many times. A resource can be an email address or an object to delete.

Every effect must satisfy two permission lists: the user's permissions and an
independent application ceiling. If the user permits two sends but the
application permits one, the second attempt is rejected. Their usage counters
are separate. Figure~\ref{fig:delivery-policy} gives the complete report policy.

\begin{figure}[htbp]
\footnotesize
\begin{verbatim}
node Records::"private" : source { sensitivity: Confidential, trusted: true }
node Directory::"records" : source { sensitivity: Public, trusted: true }
node Content::"other" : source { sensitivity: Public, trusted: false }
node Memory::"working" : internal { sensitivity: Public, trusted: true }
node Mail::"endpoint" : sink { external: true }
edge Read { from: source, to: internal, propagation: inherit(from) }
edge Lookup { from: source, to: internal, propagation: inherit(from) }
edge Other { from: source, to: internal, propagation: inherit(from) }
edge Extract { from: internal, to: internal, propagation: inherit(from) }
edge Sanitize { from: internal, to: internal, propagation: transform { set: { sensitivity: Public } } }
edge Send { from: internal, to: sink, propagation: inherit(from) }
agent assistant { max_sensitivity: Confidential
                  edges: [Read, Lookup, Other, Extract, Sanitize, Send] }
flow allow prepare { to: internal }
flow allow publish { to: sink carrying: { sensitivity: == Public } }
origin report_body { via: Send argument: body producer: Sanitize }

tool Read     { from: Records::"private", to: Memory::"working", effects: [ record.read("report-42") ] }
tool Lookup   { from: Directory::"records", to: Memory::"working",
                effects: [ directory.lookup("alice") ], authoritative: [ destination ] }
tool Other    { from: Content::"other", to: Memory::"working", effects: [ content.read("other") ] }
tool Extract  { from: Memory::"working", to: Memory::"working", arguments: [record],
                effects: [ directory.extract("local") ] }
tool Sanitize { from: Memory::"working", to: Memory::"working", arguments: [record],
                effects: [ record.release("report-42") ] }
tool Send     { from: Memory::"working", to: Mail::"endpoint", arguments: [recipient, body],
                effects: [ message.send(recipient) ] }

contract deliver_report {
  resolve target: Lookup() -> destination
  user: [ record.read("report-42"), directory.lookup("alice"), content.read("other"),
          directory.extract("local") { args: [record] }, record.release("report-42") { args: [record] },
          message.send(@target) { args: [recipient, body] } ]
  application: same
}
\end{verbatim}
\caption{Report-delivery policy used in the evaluation. The original control
omits the origin rule, and neither policy constrains \texttt{Sanitize}'s input
(Section~\ref{sec:motivation}). Every permission has quota one. The application compiles \texttt{deliver\_report} with
the authenticated request and installs exact-call approval for \texttt{Send}
separately (Section~\ref{sec:exact-release}).}
\label{fig:delivery-policy}
\Description{Report-delivery declarations: a private source, trusted public projection, authoritative recipient lookup, sanitized-body origin, public-only send, and independent quota-one user and application permissions.}
\end{figure}

\paragraph{Connecting a tool to its effects.}
The declaration for \texttt{Send} identifies its arguments and says that its
effect is \texttt{message.send(recipient)}. The monitor therefore checks the
actual address supplied as \texttt{recipient}. A tool declaration fixes its
endpoints, allowed and required argument names, and an ordered nonempty list
of effects. Resources come from constants, named arguments or list arguments.
Two recipient entries require two effect units, even if their addresses are
equal. Missing or malformed resources and undeclared argument names reject
the call. Both permission lists must cover every effect.

\subsection{Obtaining the Permitted Address}
Alice's address may be unknown when the task begins. The contract uses a
placeholder, called a \emph{slot}, and fixes how it may be filled:
\texttt{resolve target: Lookup() -> destination}.
This says to use the \texttt{destination} field from the designated
\texttt{Lookup()} result. Before that successful, authorized call, the slot
has no address and cannot authorize a send.

The rule fixes the lookup tool, exact arguments and result-field path before
execution. The tool must declare that path \kw{authoritative}. Content from
another tool cannot fill the slot merely by claiming to be a directory response.
Designating the wrong tool is instead a policy mistake that the monitor cannot
detect (Section~\ref{sec:threat}).

A repeated equal address preserves the binding. A conflicting address or
malformed successful result permanently invalidates the slot for that task.
A failed lookup does not change it. Matching compares the proposed recipient's
exact address with the stored address; the recipient handle need not itself
come from the lookup. Neither a lookup nor an instruction in its result renews
quota or changes the installed permissions.

\subsection{Matching and Charging a Permission}
For the example, \texttt{message.send(@target)} permits a send to the address
currently stored in \texttt{target}. Its argument allowlist permits
\texttt{body} and \texttt{recipient}, and its default quota permits one attempt.
Other permissions may restrict argument values or numeric ranges.

Formally, a permission is $p=(\alpha,\mathcal R,K,\chi,m)$: action, nonempty
set of resources or slot references, allowed argument names, argument
constraints (finite scalar sets or numeric intervals), and quota in effect units
($m=\infty$ means unlimited).
Let $B$ store current bindings and $I$ identify invalid slots.
$G_{B,I}(\mathcal R)$ replaces valid slot references with their addresses;
unbound or invalid references supply no resource.
Permission $p_k$ matches effect $(\alpha',r)$ with arguments $x$ when
\[
 \alpha'=\alpha_k\ \land\ r\in G_{B,I}(\mathcal R_k)
 \ \land\ \operatorname{keys}(x)\subseteq K_k\ \land\ \chi_k(x).
\]
Thus the action and resource must match, argument names must be allowed,
and argument constraints must hold.

For each effect, the monitor scans permissions in declaration order and picks
the first matching permission with capacity left. This is \emph{first-fit
allocation}. It tentatively reserves one unit, then processes the next effect.
Both sides must succeed before any charge is committed. If either fails,
all reservations are discarded. Charges commit only on Allow, before tool
invocation; tool failure does not refund them.

\paragraph{Why order matters.}\label{par:first-fit-order}
In the introduction's example, the first permission allows one deletion of B
and the second one of A, so a batch listing A then B reserves the second for A
and the first for B. Once the first also covers A, A takes it, leaving none for
B, and the whole batch is rejected before dispatch, even though the other
permission list covers both effects.

Searching for any feasible assignment could restore this batch, but would
change the allocation semantics and still allow AA unless per-resource limits
were added. The Lean comparison in Appendix~\ref{app:allocation-alternatives}
checks these cases. A stronger allocator also faces an online limitation.
The following is the two-grant instance of Barth and Mitchell's
online-authorization obstruction~\cite[Theorem~4]{barthdrm2006}.

\begin{proposition}[Unknown future calls]\label{prop:online-allocation}
With irrevocable charging, no deterministic allocator that chooses from the
observed prefix alone can complete every future client that has a feasible
allocation, for arbitrary overlapping grants.
\end{proposition}
\begin{proof}
Take quota-one grants for \{A,B\} and \{A,C\}. Both clients AB and AC have
feasible allocations. Their common first call A must charge one grant. Charging
\{A,B\} prevents the later B; charging \{A,C\} prevents the later C. The first
choice sees the same prefix in both clients.
\end{proof}
Our Lean module \texttt{AllocationAlternatives} mechanizes this known
obstruction and the feasibility claim. The result assumes unknown future calls and
no delayed charging or reassignment; it does not make first-fit optimal.
First-fit makes declaration order the rule for charging and retains each grant's
identity. Our preservation result applies to that choice of semantics.
Rejecting all overlap would also remove intentional pooled permissions.

\section{From a Proposed Call to a State Update}
\label{sec:semantics}
The monitor checks each proposed call against the installed policy and its
records of earlier calls, as Section~\ref{sec:motivation} showed for \texttt{Send}.

We now define this procedure precisely for the \emph{base monitor}.
The model assumes calls are processed one at a time, declarations remain fixed,
and rule predicates always return the same answer for the same inputs without
changing state.
Results are observed only after Allow and declared success.
Output-dependent relabeling, resets, plans and capabilities are outside this
closed model. Optional exact-call permits are added in Section~\ref{sec:exact-release}.

\subsection{The Records and the Proposed Call}
The state groups the monitor's records as $S=(M,L,\tau,B,I,u_P,u_A)$.
$M$ stores the label of each admitted attempt. $L$ records which earlier values
it used, through distinct parent IDs. $\tau$ lists the admitted attempts in
order. $B$ stores discovered resource bindings, $I$ identifies invalid slots,
and $u_P,u_A$ count use of each permission on the user and application sides.
An ID is a task-local name for a record.

The fixed environment contains the well-formed schema $\Gamma$, ordered
permission lists $P,A$, and resolver rules $R$. Registration has unique names,
valid references and labels, and satisfies Section~\ref{sec:contracts}'s tool/resolver constraints.
A fixed set $Q$ records satisfied structured requirements; changing it is
outside these transitions.

A proposed call is
$c=(a,e,n_s,n_d,J,x,h)$: the agent, operation, source and destination,
ordered distinct parent IDs $J$, concrete JSON arguments $x$, and the precise
adapter's argument-to-handle map $h$.
For \texttt{Send}, $J$ identifies the body and recipient producers,
while $x$ contains their actual values. The rules below require registered schema
references and existing normalized parents. The outcome $(ok,y)$ records success
and the returned result after authorization; authorization never depends on
the call's own outcome.

A well-formed state has distinct attempt IDs, with exactly those IDs in $M$
and $L$. Parents precede their children, so dependencies cannot form a cycle.
Labels are valid and counters stay within their finite limits. Invalid slots
have no binding, and each live binding came from an applicable successful
lookup. This last condition records the lookup's origin; correctness of the
address returned by the service remains a trust assumption.

For a source read, $J$ is empty and the input label is the declared source
label. With tracked internal inputs, their labels are combined:
$\ell_{in}=\bigsqcup_{j\in J}M(j)$. Parentless internal requests use the node
default; the precise adapter rejects this fallback if actual inputs are omitted.
The edge transfer gives $\ell_{out}=\rho_e(\ell_{in})$.
With a fresh ID $d$, the proposed record is
$\sigma=(a,e,n_s,n_d,J,x,d,\ell_{out})$.
The implementation obtains distinct IDs from append-only trace positions.
Calls rejected during preparation leave the modeled state unchanged; exceptions
are outside this returning-transition model.

\subsection{Checking, Then Committing}
Contract preparation computes tentative charges $q_P,q_A$, one counter vector
for each permission list. It checks capacity without spending it yet.
$\mathit{Struct}$ checks endpoints, agent operation membership and the output
sensitivity ceiling. $\mathit{Denied}$ means an explicit deny matches;
$\mathit{Matched}$ means at least one allow matches.
$\mathit{Ready}$ means at least one matching allow needs no requirement or has
its requirement in $Q$. One ready matching allow is sufficient.

The check $O(S,h,e)$ requires every origin rule for $e$ to hold on its named
handle (Section~\ref{sec:argument-origin}). Missing roles or handles fail it.
The additional checks $V(S,\sigma)$ and $H(S,\sigma)$ test value approval and
history, respectively. Their meanings are given by
Lemma~\ref{lem:value-approval} and Lemma~\ref{prop:path}.
Checks have the following priority:
\[
D(S,c)=\begin{cases}
\Deny & \text{preparation fails or }\neg\mathit{Struct},\\
\Deny & \mathit{Denied}\text{ or }\neg\mathit{Matched},\\
\Pause & \neg\mathit{Ready},\\
\Deny & \neg O(S,h,e),\\
\Deny & \neg V(S,\sigma),\\
\Deny & \neg H(S,\sigma),\\
\Allow & \text{otherwise.}
\end{cases}
\]
Deny blocks the call. Pause means a matching allow is waiting for a requirement;
it also blocks invocation. An explicit denial takes precedence over that wait.
In the symbolic fragment $\mathcal F$, origin, requirements, value-approval and
history rules are absent, so $\mathit{Ready}=\mathit{Matched}$ and $O=V=H=true$.
The deployed Lean core may reject a missing origin during preparation and
rechecks it at commit; either ordering prevents invocation.

Allow makes the tentative charges permanent, called \emph{commit}, and adds
the attempt record:
\[
\frac{D(S,c)=\Allow\quad d\notin\dom(M)}
 {S\xrightarrow{c/\Allow}
 (M[d\mapsto\ell_{out}],L[d\mapsto J],\tau\cdot\sigma,
 B,I,u_P+q_P,u_A+q_A)}.
\]
Here $[d\mapsto\cdot]$ adds the new record, $\tau\cdot\sigma$ appends the
attempt, and adding $q_P,q_A$ charges its reserved units.
For the first permitted send, each send counter goes from zero to one.
Only the reservations for this exact call may be committed.
Deny and Pause leave these security records unchanged; diagnostic logs are
outside the modeled state.

Commit precedes tool execution. An admitted attempt therefore has an ID even
if the tool fails. Only successful completion exposes a usable value handle,
and failure does not refund quota.

\subsection{Processing the Result}
A successful designated lookup updates $B,I$ using the cases of
Section~\ref{sec:contracts}. A failed lookup leaves them unchanged; a malformed
or conflicting successful result invalidates its slot permanently.
Inheritance needs no later label change. Runtime scanners and redaction may
finalize labels after return, but are outside this core and $\mathcal F$.
The Lean workflow model can instead take scanner observations as supplied
outcome data (Section~\ref{sec:lean-workflows}).

An approval permit checks the exact outgoing call but changes no label, so
approving a raw report cannot override the Public-only send rule.

\begin{proposition}[Determinate finite steps]\label{prop:determinate}
For fixed well-formed registration, a well-formed state, a normalized request
and a declared outcome, the closed core has a terminating, unique next decision
and state, up to consistent renaming of fresh IDs.
\end{proposition}
\begin{proof}
Every traversal is over finite parents, rules, effects or permission lists.
Total pure predicates and fixed transfers determine their results; first-fit
selects a unique reservation, and the ordered decision cases are exhaustive.
Only Allow commits; the fixed outcome selects exactly one observation case.
Freshness leaves only the choice of an otherwise unobservable new name.
\end{proof}

\section{What the Language Guarantees}
\label{sec:properties}
The results address client preservation, declared argument origins, private
data flow and task authority, each under the stated
application and model assumptions.
\subsection{Preserving Clients Across Policy Revisions}
\label{sec:client-theorem}
A \emph{client trace} fixes a finite sequence of calls, effects,
service results and resolver observations. Let $E$ contain the declared effects
with arguments, and let $\mathrm{Adm}_E(P,A)$ be the client traces with effects
in $E$ that the monitor admits from a fresh task under user and application
ledgers $P$ and $A$. A revision to $(P',A')$ \emph{preserves clients on $E$} when
\[
  \mathrm{Adm}_E(P,A)\ \subseteq\ \mathrm{Adm}_E(P',A').
\]
Preservation does not predict future calls or service results. First consider
three structural revisions, composed on both ledgers.

A \emph{prefix-compatible extension} retains quotas, usage and matching
predicates at old positions, with arbitrary grants appended. Matching must agree
for every $e\in E$ and binding state. The old allocation then charges the same
accounts. Appending unchanged grants qualifies for all effects, without a finite
domain; matching-equivalent edits can qualify on a restricted $E$.

A \emph{capacity extension}, such as a quota increase, retains matching and
order, but gives each old
position at least its old residual capacity: $m-u$ for finite quota $m$ and usage
$u$, or $\infty$ for an unlimited quota. It may also append grants. Counters
and allocation choices need not remain equal. Increasing the report's lookup
quota supports address revalidation while its send quota can remain one.

Two accounts are \emph{independent on $E$} if no $e\in E$ matches both under any
binding state. This ignores remaining capacity and therefore survives charges.
Let $P\rightsquigarrow_E P'$ be generated by the two extensions, lifting under
an unchanged head, adjacent independent swaps and finite composition.
Only accounts that cross must be independent; overlapping grants can retain
their relative order.

\begin{theorem}[Client preservation under revision]
\label{thm:client-preservation}\label{thm:independent-revision}
Suppose $P\rightsquigarrow_E P'$, $A\rightsquigarrow_E A'$ and bindings agree.
Under first-fit allocation, (1) the revision preserves clients on $E$: every
admitted old client remains admitted at every finite length, with related
successor ledgers and equal bindings; and (2) starting from identical safety
states, for serialized calls under a fixed strict-monitor configuration,
unchanged tools and resolver rules, and identical service, scanner and reviewer
observations, it also preserves full safety state and every admission of the
deployed monitor. Admitted service failures retain their charges.
\end{theorem}
\begin{proof}
For a capacity extension, if the old allocation selects $j$, the new selects
$k\leq j$. When $k=j$, charging both preserves residual-capacity dominance.
When $k<j$, the matching old grant at $k$ is exhausted, so charging the new grant
still leaves at least zero there; charging only the old grant at $j$ preserves
dominance there too. An independent swap preserves the charged account under
its position permutation. Compose these simulations on both ledgers. Equal observations preserve
bindings; failure retains charges. For the full monitor, equal safety states
produce equal flow, origin, history, value-approval and exact-permit decisions.
Induction over commits, scans, completions and reviewer events preserves this equality.
\end{proof}
In Lean, \texttt{CapacityPreservation} proves the capacity case and
\texttt{IndependentRevision} composes the steps.
The introductory widening fails the criterion. Conversely, appending another A
grant preserves old clients but admits AA: safety remains a separate obligation.

\paragraph{Keeping intentional overlap.}
Consider the ordered quota-one grants
$[\{A,B\},\{A,C\},\{D\}]$. Revise them to
$[\{D\},\{A,B\}^{2},\{A,C\}]$, where the superscript increases capacity to two.
Only D crosses another grant; the overlapping grants stay in order. Lean accepts
the two swaps that move D forward and the capacity increase, and rejects a swap
of the overlapping pair (\texttt{RevisionSeparation}). Because neither
overlapping resource set contains the other, the check keeps intentional overlap
while certifying this particular revision. It needs neither a globally
well-behaved license class nor reassignment of earlier charges, the alternatives
studied by Barth and Mitchell~\cite{barthdrm2006}.

\paragraph{Structural certificates.}
For a fresh registration, Python proposes a derivation of $\rightsquigarrow_E$,
a sequence of independent swaps followed by capacity extensions. Lean recomputes
the matching tables with the deployed matching predicate and checks the
derivation against them, so the proposal need not be trusted
(\texttt{RevisionCertificate}).

\paragraph{Simulation certificates.}
A finite-quota revision without a structural derivation can be certified by a
simulation instead. Python proposes a relation between old and new counters, and
Lean checks that it contains the initial pair and is closed under every
binding/effect row: each successful old charge has a successful new charge whose
successor pair is again related (\texttt{FiniteLedgerSimulation}). Induction then
gives preservation at every finite length. Exhausted budgets are inconclusive,
and because closure is checked per effect, a preserving revision can be rejected
when atomic multi-effect calls never reach some intermediate counters. Either kind
of certificate carries through preparation, commit and completion
(\texttt{RuntimeRevision}).

\paragraph{Which requests a certificate covers.}
Certificates are checked on finitely many matching-table rows, so each request
needs a representative row. Lean shows that one exists in every resolver state:
for an effect on resource $r$, keep the slot bindings equal to $r$ and treat all
others as unbound (\texttt{RegisteredCoverage}). This preserves every grant's
match, including aliases, invalid slots and resources outside the inventory, so
no bound on reachable bindings is needed. Payloads need not be enumerated either:
the proof quantifies over arbitrary JSON in fields that no ledger constraint or
resource expression inspects, on all four ledgers (user and application, old and
new) (\texttt{RequestFamilyCoverage}). The report can thus vary
\texttt{Send.body}, \texttt{Extract.record} and \texttt{Sanitize.record}, while
recipients stay fixed.

\paragraph{Lifting to the deployed monitor.}
The certificates also lift through the deployed core and safety functions, which
gives part (2) of Theorem~\ref{thm:client-preservation} (\texttt{EngineRevision}).
Calls are serialized: each prepares, commits, receives its scans and completes
before the next, with approval and requirement events in between and failed
services included. Interleaved outstanding calls, resets and post-denial states
are excluded. A rejected certificate is inconclusive; it does not by itself show
that a client is lost.

\subsection{Argument Origins at Admission}
\label{sec:argument-origin}
A label describes classification, not which argument was produced by which
operation. For a producer $p$ and declared preserving edges $K$, define
$\mathit{Origin}_{p,K}(j)$ on recorded trace positions by two rules:
\[
\frac{\mathit{edge}(j)=p}{\mathit{Origin}_{p,K}(j)}
\qquad
\frac{\mathit{edge}(j)\in K\quad\mathit{parents}(j)=\{i\}\quad i<j
      \quad\mathit{Origin}_{p,K}(i)}{\mathit{Origin}_{p,K}(j)}.
\]
Parents are distinct tracked identities. With $K=\varnothing$, only an immediate
$p$ output qualifies. A declared preserving operation may have repeated aliases
of one parent, but a merge of different parents does not inherit origin.

\begin{theorem}[Declared argument origin]\label{thm:argument-origin}
For an applicable origin rule on a successfully prepared or committed call,
the declared argument's authenticated value handle resolves to a trace position
$j$ satisfying $\mathit{Origin}_{p,K}(j)$. The executable origin test is true
exactly when this inductive relation holds.
\end{theorem}
\begin{proof}
The test follows strictly earlier parent positions and agrees with the two
rules by induction. Preparation and commit require that test on the selected
argument's stored handle. Missing roles or handles fail the check.
\end{proof}
Lean proves both the correspondence and its use at commit and preparation in the
deployed core (\path{Provenance.originAt_iff},
\path{Safety.committed_argument_has_origin}). Equal payloads never confer origin.
The application still chooses the producer and preserving operations, and the
theorem does not show that \texttt{Sanitize} removes every secret.

\subsection{Value Flow and Task Authority}
In the report task, copying or combining the private report must retain its
Confidential classification. Only a trusted release, such as the report
projection, may lower it. A \emph{release boundary} marks such a transfer.
An \emph{uncut ancestor} is a source connected to the current value without
passing through a release boundary. Other transfers must preserve or increase
joined sensitivity; a constant transfer must be an independent source or a
trusted release.

\begin{theorem}[Confidential ancestry exclusion]\label{thm:causal-secrecy}
Start with a well-formed empty task. Suppose every actual dependency is tracked,
source classifications are faithful, only trusted release boundaries may lower
sensitivity, and every external send requires a Public joined argument label.
Then no admitted external send has an uncut Confidential source ancestor.
\end{theorem}
\begin{proof}
Induct over the acyclic value dependency graph. A source has its declared
sensitivity. At a non-release node, the join dominates each parent's sensitivity,
and transfer cannot decrease it. Thus every uncut source ancestor's sensitivity
is at most the node's sensitivity. A release cuts only paths through its output;
it leaves other paths and the original source unchanged. Apply this invariant
to the joined arguments of a send. A remaining Confidential ancestor contradicts
the required Public label.
\end{proof}
Thus combining the Public summary with the original report blocks the send.
The deployed runtime satisfies a restricted form of the theorem:
Lean relates stored values and trace positions to unfolded causal trees through
the deployed preparation, commit, scan, completion and approval steps, and
instantiates the theorem for ordinary transfers that preserve sensitivity exactly
(\texttt{RuntimeCausal}). Checked predicates establish such transfers and Public flow
guards for the report registration; general registrations need not satisfy them.
The result assumes parentless source calls, complete dependencies and faithful
labels. Trusted declassifiers cut ancestry; exact-call approval does not. Timing, hidden reads and faulty
projections are outside this guarantee; it is not a noninterference theorem.

The contract also limits which effects can occur. In the example, every send
needs permission from both sides, and its permitted address must come from the
designated lookup. Invalidating that address prevents the slot from being reused.

\begin{theorem}[Task authority and resource origin]\label{thm:contract}
Start from a fresh well-formed task state: finite counters are within quota and
all resource slots are unbound, and assume that only the designated resolver
produces resolver observations. Under the contract transitions, both permission
lists cover each admitted effect under the bindings at admission. Each finite per-permission counter stays within its own quota;
quota is not a sum over the whole list. Each live slot originated in a designated
successful lookup observation, and an invalidated slot never becomes live again
during that task.
\end{theorem}
\begin{proof}
Induct over admitted attempts and observations. First-fit preparation checks every
expanded effect and both capacity vectors before commit. No failed preparation
commits partial reservations. Only the designated observation transition can
create a binding; conflict and malformed results remove it, and invalidity is
absorbing. Neither dispatch nor observation changes the permission lists.
\end{proof}

Appendix~\ref{app:proofs} states the admission invariant and, separately, the
exact-call dispatch invariant; Appendix~\ref{app:lean-core} maps the preservation
proofs to Lean. The results in this section concern the base
monitor. We next relate that monitor to bounded analysis;
Section~\ref{sec:exact-release} then composes the optional permit extension.

\section{Bounded Workflow Analysis}
\label{sec:analysis}
When neither preservation check applies, the analyzer explores finite
sets of calls and results up to a chosen length. Its primary implementation
uses the deployed Lean transitions; the search proof covers the supplied finite
model. A second implementation translates a restricted fragment to Z3 constraints
and needs a separate correspondence argument. Both depend on the application's
coverage of calls and outcomes. We replay candidate violations on tools.
Table~\ref{tab:proof-status} summarizes the proof boundaries.

\begin{table}[h!]
\caption{Claims, evidence and remaining premises. All rows retain trusted
Python translation, transport and external tool behavior; none verifies the
complete integration.}
\label{tab:proof-status}
\centering\footnotesize
\begin{tabular}{p{0.22\linewidth}p{0.33\linewidth}p{0.33\linewidth}}
\toprule
Claim & Code, proof and application evidence & Remaining boundary \\
\midrule
Revision preserves clients & Structural or finite-ledger certificates; payload families and full monitor simulation & Fresh serialized tasks, covered requests, fixed configuration and identical observations \\
Arguments have declared origins & Deployed check, inductive correspondence and commit theorem; report delivery & Complete dependencies and trusted producer/preserver behavior \\
Runtime respects rules and authority & Selected deployed allocation/lookup/permit proofs; flow/history correspondence & Primitive predicates, trusted resolver preparation and complete mediation \\
Public sends exclude uncut private ancestry & Deployed-state induction for supported transfers and Public flow guards & Parentless sources, faithful labels, complete dependencies and trusted declassifiers \\
Bounded workflow search covers its model & Shared Lean transitions, proved search and delivery query; application replay & Finite outcome/reviewer coverage, scanner observations and explicit bound \\
Finite formulas describe the restricted monitor & $\mathcal F$ manuscript proof, Lean contract abstraction and finite certificates & Abstracted noncontract checks, syntax export, replay tables and model absence; no universal compiler proof \\
\bottomrule
\end{tabular}
\end{table}

\subsection{Exploring Workflows with the Runtime's Lean Functions}
\label{sec:lean-workflows}
The analyzer, \texttt{agenra-analyze}, explores workflows by calling the deployed
Lean preparation, commit, scan and completion functions, with supplied results in
place of real tool execution. From an immutable strict-profile registration, it
enumerates ordered handle assignments (including aliases), permit choices,
declared outcomes and enabled reviewer events. Each admitted call or reviewer
update is one event; denials commit nothing, and an admitted failure is charged
and ends the branch.

\paragraph{What the search guarantees.}
Let $N_M$ be the successor function of the supplied finite outcome and reviewer model
$M$ and $s_0$ its initial state; nodes hold core state, history and the latest
call observation. Lean proves the search sound and complete against an
independently defined at-most-$k$ reachability relation: for a completed
budgeted result $R$,
\[
 s\in R\quad\Longleftrightarrow\quad
 \operatorname{Reach}^{\leq k}_{N_M}(s_0,s).
\]
Hence if search finds no Bad node, no reachable node within the bound is Bad,
and every reported violation is reachable in the model. A streaming variant keeps
only a depth-first frontier and query witnesses; we prove that its summary equals
a fold over the whole tree, so both results carry over (\texttt{Workflow},
\texttt{StreamingSearch}).
Streaming reduces memory use but still visits every node, since it neither merges
states nor drops history.

\paragraph{Queries and verdicts.}
Queries are Boolean combinations of label tests, typed argument equality,
argument origins, admitted-call counts and permit presence. A Bad node is a
modeled violation. Without one, a Useful witness yields \path{BOUNDED_SAFE}
and no witness yields \texttt{VACUOUS}; unsupported inputs yield
\texttt{UNSUPPORTED}, and exhausted budgets, analysis errors and timeouts yield
\texttt{INCONCLUSIVE}. Here Useful refers to the supplied predicate; it does not
mean the task eventually completes.

\paragraph{The delivery predicate.}
For an observed \texttt{Send}, the report model requires the body's inductive
\texttt{Sanitize} origin (not the projection's input; Section~\ref{sec:motivation}),
typed recipient equality with Alice's address, at most
one admitted \texttt{Send} event, and a supplied permit. Failed admitted sends
count too. We prove that the executable Bad query is exactly the negation of
this independently defined conjunction, so a complete search with no Bad node
establishes the property for every reachable modeled Send within the bound
(\texttt{DeliveryApplication}). The runtime permit
invariant separately supplies exact issuance and consumption. This ties the
search result to an application-level property; coverage of the actual tools and
reviewer remains an assumption, which Section~\ref{sec:eval} evaluates for a
fixed local service.

\subsection{Complementary Finite-Formula Analysis}
\label{sec:encoding}
The complementary Z3 analyzer translates a restricted monitor. Finite domains
$D_i$ list candidate calls, earlier producers, typed arguments and outcomes at
each position. The formula $\Phi_D$ selects candidates and constrains labels,
counters, bindings and decisions. Every intermediate decision must Allow; the
last may be denied. Analysis checks every prefix length, including admitted
terminal failures.

\paragraph{Scope and correspondence.}
The fragment $\mathcal F$ fixes a well-formed registration, empty initial state,
serialized calls, finite typed JSON arguments/outcomes and earlier parents.
Sources have constant labels, and transfers either inherit or apply, in order,
category removal, sensitivity decrement, sensitivity assignment and trust assignment.
Structural and flow checks are pure and total on the finite label universe;
contracts follow Section~\ref{sec:contracts}. Six category paths bound predicate
tabulation at $10\cdot2^6=640$ labels, a limit of our construction rather than
of the semantics.
History, argument origins, release, approvals, capabilities, resets, requirements,
scans and output-dependent or side-channel transfers are excluded. Only admitted
successful designated results trigger observations. The report policy without
origin rules or permits lies in $\mathcal F$.

Relate concrete state $S$ to encoded state $\widehat S$ by mapping committed IDs
to producer positions, preserving parents and labels, and equating counters,
invalidity and bindings through an injective resource encoding.
\begin{theorem}[Finite correspondence for $\mathcal F$]\label{thm:correspondence}
For each well-formed selection from the declared domains, the selection-fixed
formula has a model iff the concrete authorization prefix is admissible, meaning
all intermediate decisions are Allow. Every such model agrees with all decisions,
admitted labels, slot bindings, invalidity and counters.
\end{theorem}
State induction proves this manuscript theorem (Appendix~\ref{app:correspondence}).
Lean mechanizes the contract part, treating the other checks as one abstract
readiness predicate. Finite certificates also check selected emitted equations
and complete formulas against fresh runtime replay tables; they do not prove the
compiler correct for every input.

\paragraph{Coverage and refinement.}
The restricted adapter makes serial JSON calls with immutable whole handles
and registered literals. Preparation denials preserve security state;
admitted failure is charged and terminal. The finite domain builder enumerates
all earlier-handle role assignments, including aliases, and supplied outcomes.
\begin{proposition}[Conditional adapter inclusion]\label{prop:adapter-inclusion}
If construction completes and the supplied outcomes cover every possible
successful result and terminal exception, its domains contain every projected
restricted-adapter execution up to the requested bound. Projection removes
state-preserving denials.
\end{proposition}
Each handle has an earlier successful producer and each outcome is covered.
Two refinements exclude invented data: $C_{\rm value}$ requires whole-producer
values, exact registered roles/literals and the ordered deduplicated parent
list; $C_{\rm outcome}$ requires call-specific success/result membership.
Let $\Psi_D=\Phi_D\land C_{\rm value}\land C_{\rm outcome}$.
For covered executions $Exec_D$, the needed relation is
\[
 \operatorname{Proj}(Exec_D)\subseteq\operatorname{Traces}(\Psi_D)
 \subseteq\operatorname{Traces}(\Phi_D).
\]
Only the second inclusion follows from conjunction alone. The first needs
coverage of all reachable outcomes, including failures; per-call outcome
membership still need not imply joint service realizability.

\begin{theorem}[Bounded exclusion for the restricted executor]\label{thm:bounded-exclusion}
Assume $\mathcal F$ and Proposition~\ref{prop:adapter-inclusion}'s hypotheses.
Let Bad depend only on selected invocations and related observations, be false
initially, and be unchanged by removing state-preserving denials. If the Bad
query for $\Phi_D$ is UNSAT for every prefix length $1,\ldots,b$, no projected
executor prefix of at most $b$ admitted attempts violates Bad.
\end{theorem}
\begin{corollary}[Refined bounded exclusion]\label{cor:refined-exclusion}
The same conclusion holds for $\Psi_D$ if the displayed execution inclusion
holds at every checked length.
\end{corollary}
A violating execution would give a Bad model by inclusion and correspondence,
contradicting UNSAT. Unsupported inputs, incomplete construction and UNKNOWN
give no exclusion. Appendix~\ref{app:adapter-projection} supplies the argument.

Both analyzers check Bad and Useful separately, and a Useful witness shows only
one completion. A witness is \emph{executor-confirmed} when replaying its exact
arguments on the tools reproduces its results. Disagreement is inconclusive, and
agreement confirms that execution without showing the service model complete.

\section{Optional Exact-Call Approval}
\label{sec:exact-release}
\paragraph{Approving one specific call.}
A reviewer may approve \texttt{Send(body=s, recipient=a)} for one use.
The adapter derives arguments and producer identities from immutable handles
or fixed literals. The monitor stores issued approvals in $\Pi$ and spent
approvals in $U$. Each permit's signature contains the operation, ordered
argument-role/value-ID pairs and typed concrete arguments. Equal content under
a new handle requires fresh approval; other policy checks still apply.

Trusted Issue adds a fresh permit to $\Pi$. A protected call needs a matching
permit in $\Pi\setminus U$; a mismatch rejects it before base authorization.
If the base monitor returns Deny or Pause, $U$ is unchanged; on Allow, the permit
moves to $U$ before invocation. A failed tool keeps both the quota charge and the
consumed permit. Issuance and dispatch are serialized.

\begin{theorem}[Exact-call dispatch]\label{thm:permit}
Starting from $\Pi=U=\emptyset$, consider these permit transitions composed with
the core within one task, under trusted issuance, immutable registered values
and literals, exact signature equality, complete mediation and serialized
transitions. Every protected dispatch has passed base Allow and
consumed an unused permit whose signature matches its exact operation and
arguments and was issued earlier in the same task. Each permit is consumed by
at most one local dispatch, including a dispatch whose tool fails.
\end{theorem}
The proof is an induction over the permit transitions: Issue only adds fresh
identities, only a base Allow moves a matching permit into $U$ before invocation,
and no later step removes a permit from $U$ (Appendix~\ref{app:proofs}). The Python reference backend
compares canonical digests and assumes collision freedom; deployed Lean compares
structures. The contract's one-send quota and each permit's one-use consent
remain separate checks.

\paragraph{Modeling approval in bounded analysis.}\label{sec:permit-encoding}
The shared Lean analyzer explores the same issuance and consumption transitions,
counting tool and reviewer events separately (Section~\ref{sec:lean-workflows}).
The Z3 extension instead models each scheduled approval offer with two bits,
issued and consumed; an offer may be issued before any candidate whose producers
have succeeded. Base Allow commits quota and consumption together, and an
admitted failure is terminal. The bound counts tool candidates,
with at most one Issue per offer; arbitrary approval arrival and resets are not covered.
Appendix~\ref{app:permit-symbolic} states its conditional correspondence and inclusion.

\begin{corollary}[Authorized delivery dispatch]\label{cor:delivery-lifecycle}
Consider a fresh task mediated by the exact-call adapter under the declarations
of Figure~\ref{fig:delivery-policy}. Assume complete tracking of actual argument
dependencies, faithful source classifications, lowering only at trusted release
boundaries, a sound public projection, and an authoritative designated directory;
assume the remaining hypotheses of Theorems~\ref{thm:causal-secrecy},
\ref{thm:enforcement}, \ref{thm:argument-origin}, \ref{thm:contract}
and~\ref{thm:permit}. Every protected \texttt{Send} dispatch admitted by the
adapter uses a Public body with declared \texttt{Sanitize} origin, the
resource currently bound by the designated resolver, and a matching unused permit
consumed before the registered tool is invoked. At most one such \texttt{Send}
is dispatched during the task, including a dispatch whose tool fails.
Lookup invalidation and failed dispatch preserve this claim.
\end{corollary}
\begin{proof}
Compose confidentiality, argument-origin admission, designated resource origin
and exact permit consumption. Each \texttt{Send} consumes the sole quota-one
\texttt{message.send} permission on each ledger; failure does not restore it.
Invalidation prevents reuse of the binding. Each invariant holds at the dispatch point.
\end{proof}

\section{Implementation}
\label{sec:implementation}
Python translates policies, forwards calls to Lean, invokes tools and reports
results. It also implements finite Z3 analysis~\cite{z3} and replay. The persistent
Lean process owns security state and processes results before returning handles
(Figure~\ref{fig:architecture}).

An instance lock prevents overlapping or reentrant execution until result
processing finishes; approval issuance uses the same lock.
Failed admitted calls retain charges and close the precise execution prefix.
The strict profile rejects rules needing Python extension guards.
Compatibility mode retains such guards for callbacks and advanced features.
The Python reference backend supports conformance checks and cross-backend replay.
Framework callbacks must propagate rejection so the tool is not invoked
(Appendix~\ref{app:implementation-details}).

The deployed history matcher is proved sound and complete against an independent
interval relation, including nested sequences, choice, repetitions and negation.
Only advancing steps count toward repetition, so history bounds enumeration.
The runtime uses nesting fuel 256 (Appendix~\ref{app:lean-core}).

Policy translation, decoding, primitive predicates, SQL queries, transport,
compiler/runtime behavior and scanner judgments remain trusted; Python must invoke
exactly the operation and arguments Lean checked (Appendix~\ref{app:lean-core}).

\section{Evaluation}
\label{sec:eval}
We evaluate whether Agenra detects unsafe revisions and lost workflows (RQ1),
whether its analyses agree with tool execution (RQ2), how much analysis costs
(RQ3), and how agents behave on clean and attacked tasks (RQ4).
We report a violation only after replaying its witness on the tools.
A bounded absence result depends on the relevant analysis's
domain and outcome-coverage assumptions (Section~\ref{sec:analysis}).
An episode is one fresh task run; an attack chain groups three attempts.
Appendix~\ref{app:evaluation-details} records coverage and protocol details.

\begin{table}[t]
\caption{Policy-revision controls. Containment checks individual effect coverage;
BA checks one batch ordering; Clients counts retained declared workflows.
Prefix checks prefix-compatible extension alone; Checked adds capacity extensions
and independent swaps on the declared domain. -- means inconclusive.
Bad denotes a replayed forbidden trace.}
\label{tab:revision-clients}
\centering\small
\begin{tabular}{lccrccc}
\toprule
Revision & Containment & BA & Clients & Prefix & Checked & Bad \\
\midrule
Unchanged & Y & Y & 6/6 & Y & Y & N \\
Widen first B to AB & Y & Y & 4/6 & -- & -- & Y \\
Append A grant & Y & Y & 6/6 & Y & Y & Y \\
Remove A grant & N & N & 1/6 & -- & -- & N \\
Increase B quota & Y & Y & 6/6 & -- & Y & Y \\
Widen last A to AC & Y & Y & 6/6 & Y & Y & Y \\
Append C grant & Y & Y & 6/6 & Y & Y & Y \\
Reverse disjoint grants & Y & Y & 6/6 & -- & Y & N \\
\bottomrule
\end{tabular}
\end{table}

\subsection{RQ1: Policy Changes Can Lose Workflows}
Per-request checks and single successful runs both miss lost workflows. The
certificates separate preserving from regressing revisions in our corpus, and
preservation and safety diverge in both directions. The six required clients cover A, B, AB and BA batches, plus separate-call AB
and BA workflows. For eight revisions, Table~\ref{tab:revision-clients} compares
per-effect coverage, a single BA completion, the sufficient preservation
condition, all six clients and forbidden traces.
Prefix search and the symbolic compiler agree on all 48 client queries;
fresh SQLite executions confirm their effects.

Widening the first B grant to AB passes both per-effect containment and the BA
run, yet loses both AB clients and admits a replayed bad trace
(Table~\ref{tab:revision-clients}). Conversely, four changed policies satisfy the checked
preservation condition while admitting bad traces. Preservation and safety are
therefore independent obligations. A branching-deletion study shows the same gap:
with overlapping grants, all three benign manifests admit some completion but
only two keep the standard workflow. A batch tool registered with fewer effects
than it performs lets the monitor admit an undeclared deletion. Analysis of
67 systematic policy mutants finds five with replayed over-permission, 22 with
lost utility, two with both, 31 with unchanged existential verdicts and seven
rejected at parse time.

\paragraph{A published license revision.}
Barth and Mitchell's OMA example has one play from \{a,b,c\} and one dated
play of d~\cite{barthdrm2006}. Widening the dated license to \{a,d\} loses
client ad under OMA's dated-right priority~\cite[Section~5.9]{omarel2008}.
Over the 21 words of length at most two, the widening raises the number of
admitted traces from 11 to 15 but keeps only 10 of the 11 old ones, and the
checker does not certify it. Replays on both backends agree with separately
written rule equations. We model the example at a fixed time and run no native DRM software.

\paragraph{A source-defined maintenance change.}
Dyad PR~3147 raises the Basic Agent free-message quota from five to ten per
window~\cite{dyadquota2026}. We froze its source and protocol before encoding
the change. The checker certifies the literal increase. On a two-chat alphabet,
all 63 old admitted traces remain among the 2,047 revised traces, counting the
empty trace and prefixes. A separate task sends batches of 1--11 messages;
batches of up to five complete under the old quota and up to ten under the new.
Both backends agree with independent counter equations and reopened SQLite
effects on 4,272 replays. We also ran Dyad's pinned quota, mark and unmark
functions against database and clock test doubles. They confirm the same limits
for successful messages but show two differences: Dyad refunds quota when a
message fails and expires its windows, whereas our model keeps failed-attempt
charges and has no timed reset. The result therefore covers serial successful
messages in a fresh window, not equivalence of the two applications. The change
comes from Dyad, but we wrote its translation and the study.

\paragraph{Preserving intentional pool overlap.}
A second case adapts FlexNet's ordered Engineering/Sales pools with overlapping
access~\cite{flexnetallocation2026,flexnetaccess2026}. Our revision raises
Engineering capacity and moves an independent audit grant while keeping the
overlapping pools in order, and the checker accepts it over 448 matching rows.
Enumeration confirms that all 101 old traces survive with distinct directory
results, and all 79 in an aliasing control (counting empty and incomplete
prefixes), and replays on both backends agree with separate allocation equations.
Removing a
grant loses clients, and the checker cannot certify a preserving reorder of the
overlapping pools. We wrote the directory, quotas and change ourselves; the
adaptation does not model native FlexNet, lease returns, renewal or migration.

\paragraph{Coverage of the preservation checks.}
Our revision matrix contains all 24 reorderings of a four-grant policy, plus
append, removal, quota, resource, argument and resolver edits, for 51 revisions
in total. An independently implemented product-state oracle, with its own
matching and allocation equations, explores every admitted word over the
declared call and result alphabet; finite quotas make it terminate. Structural certificates cover
16 of 21 preserving revisions. Checked finite-ledger simulations certify the other
five: two resource widenings, an argument-allowlist widening, and the removal and
the quota reduction of an unreachable grant. Neither check certifies any of the 30 regressions;
each lost-client witness replays on both backends. Both resource widenings still
violate the separate safety specification, and a further control with a
one-state budget is reported as unresolved. Simulations thus close the five gaps
in this corpus, which is neither evidence of completeness nor a representative
sample of maintenance changes.

\paragraph{Maintaining the report application.}
Increasing the lookup quotas in Figure~\ref{fig:delivery-policy} to two
certifies preservation and enables address revalidation. Increasing the send
quotas also certifies preservation, but allows a second freshly approved send.
We ran five clients (delivery, revalidation, a second approval, a reused
approval and a raw body) under three policies on both backends, 30 executions
in all. The JSON/SQLite application keeps its origin, discovered-recipient and
approval checks, and its durable effects and permit use match each revision.

\subsection{RQ2--RQ3: Agreement with Execution and Analysis Cost}
Every witness the analyzers report replays on the tools, and checking a revision
takes about a tenth of a second, but exhaustive workflow search grows quickly
with the bound.
\paragraph{Structural revision-checking cost.}
We timed fifteen configurations, three runs each, varying grants (4--128),
samples (4--128), overlap and referenced slots (0--5). Twelve certify, one
overlapping reversal is inconclusive, and the row-cap and deadline controls fail
as intended. Accepted cases have 4--5,832 rows; median candidate generation takes
0.30--93.15\,ms, and median checking, including serialization, startup and
hashing, takes 69.6--141.6\,ms. The largest proposal cost is a 64-grant independent reversal.
For $d$ samples, $r$ distinct resources and $s$ slots, enumeration has
$d(r+1)^s$ rows; the five-slot configuration requires 472,392, exceeding the
200,000-row cap.

\paragraph{Combining the delivery checks.}
The report model combines argument origin, designated recipient, quota and
exact-call approval in the shared Lean transitions, with the Bad predicate of
Section~\ref{sec:lean-workflows}. Its outcomes are the actual JSON/SQLite outputs
under every handle assignment, including schema errors and a declared send
failure before insertion. The modeled reviewer approves any proposed Send, so
the origin rule alone must reject bodies from other producers.

With all six tools, the model is vacuous at bounds four and five and useful at
six; materialized enumeration exhausts its 250,000-node budget at seven and
eight. Table~\ref{tab:full-interface-search} shows the streaming runs, which
complete bound eight within a 50-million-node budget. With quota one the model
is bounded-safe. Raising only the send quota to two yields a repeated-send
witness with distinct fresh approvals and sanitized bodies, which replays on both
backends. Because streaming still visits the entire tree, the jump in time from
six to eight events limits this approach to small workflows.

\newcommand{\FullInterfaceSearchRows}{%
Materialized & 1 & 6 & 66,575 & 0.819 & 118.3 \\
Streaming & 1 & 6 & 66,575 & 0.793 & 34.8 \\
Streaming & 1 & 8 & 12,334,711 & 154.105 & 34.9 \\
Streaming & 2 & 8 & 12,334,747 & 158.935 & 34.8 \\
}

\begin{table}[t]
\caption{Full six-tool report interface. Each row is one isolated local run,
including startup and transport in search time. Peak child memory includes the
analyzer and replay cores, excluding the parent Python process. Quota-one rows
are safe and useful within the bound; quota two has a replayed second-send violation.}
\label{tab:full-interface-search}
\centering\small
\begin{tabular}{lrrrrr}
\toprule
Traversal & Quota & Bound & Node visits & Search (s) & Child MiB \\
\midrule
\FullInterfaceSearchRows
\bottomrule
\end{tabular}
\end{table}

Removing the origin rule, widening the recipient or removing approval each
yields a replayed violation. On both backends, reused approvals are rejected,
failed sends keep their charges and end the prefix, all six legitimate delivery
orders succeed, and the five frozen body-origin proposals are rejected.

\paragraph{Matched CaMeL policy interface.}\label{sec:camel-comparison}
To see where the report's enforcement state lives, we also implemented the task
on CaMeL's unmodified policy engine and value/capability classes (commit
\texttt{f083b6b}). We wrote both adapters to enforce the same requirements,
adding explicit recipient binding, quota and approval bookkeeping for CaMeL;
CaMeL's planner and interpreter are not used. Starting from empty tasks, 384 lifecycle cases
vary body form, copies, resolution, approval and send/retry schedules; 32
JSON/SQLite cases cover maintenance revisions and omission controls. Both
backends agree with the finite requirement oracle on every case and detect
all six maintenance omissions.

The Agenra adapter uses 79 logical lines plus 37 declaration lines, versus
CaMeL's 70, but three branches and one mutable field versus 13 and nine.
Agenra moves binding and quota bookkeeping into declarations that revision
analysis can check. The counts come from an earlier snapshot of both adapters
for the original policy; they measure code size, not authoring effort.

\paragraph{Supporting local models.}
Table~\ref{tab:lean-applications} covers three smaller models: payment and
assignment fixtures whose effects append to lists, and a JSON/SQLite delivery
application; their outcomes include schema errors. Eighteen queries across three
bounds and two policies yield five safe, eight violating and five vacuous results.
Every returned witness replays on both backends.

\newcommand{\LeanApplicationRows}{%
Report delivery & 5 & 416 & 452 & 85.8 & 85.3 \\
Supplier payment & 4 & 54 & 78 & 81.0 & 83.9 \\
Issue assignment & 4 & 54 & 78 & 82.1 & 82.6 \\
}

\begin{table}[t]
\caption{Supporting local models, separate from the combined delivery study.
At the shown bounds, base models are useful and bounded-safe; weakened models
have replayed violations. Nodes count explored states. Times are medians of three
fresh processes, including serialization, hashing and transport, excluding model
construction and replay. The weakened payment and assignment models widen grants;
the weakened delivery model removes the original policy's sink guard.}
\label{tab:lean-applications}
\centering\small
\begin{tabular}{lrrrrr}
\toprule
Application & Bound & \multicolumn{2}{c}{Nodes} & \multicolumn{2}{c}{Analysis (ms)} \\
 & & Base & Weakened & Base & Weakened \\
\midrule
\LeanApplicationRows
\bottomrule
\end{tabular}
\end{table}

Separate runtime measurements use 30 paired tasks per application after
3 warmup pairs, alternating backend order. Median Lean/Python workflow times are
1.27--3.04/0.76--1.52\,ms for 3--5 calls (1.67--2.01$\times$).
Lean task registration costs 22.0--22.3\,ms versus 0.050--0.052\,ms in Python.
Workflow timing includes actual tools and persistent-process IPC; registration
is separate. These are local measurements, not scalability or model-latency claims.

A 24-case sweep over source count and bound with a 2,000-node budget yields nine
safe, eight vacuous and seven exhausted cases: domains with four sources exhaust
the budget from bound six, and those with eight sources from bound five.

\paragraph{Symbolic and prefix analyses.}
In earlier runs, the symbolic approval lifecycle agreed with 52 adapter
selections, and eight JSON/SQLite maintenance protocols replayed with matching
arguments, results and effects. SMT and prefix search agreed on primary verdicts;
prefix search was faster in all 126 application runs, while on 108 monitor-only
runs neither was consistently faster.

\subsection{RQ4: Agent Execution Studies}
Both guards stop the fixed raw-send attacks, but attacks lower task completion in
every arm, and no live attack exercised the gap that the origin rule closes.
We ran these agent trials before adding the origin rule, so they use earlier
policies. Each trial compares three arms on the
same services: no guard, a direct guard we wrote by hand, and Agenra. Our first
protocol failed even clean delivery; a follow-up with typed steps and roles yields 5/5 deliveries per arm for Qwen and Hermes, while
reproducing the overlapping-AB loss under Agenra.

A separate 120-episode schedule per model uses the original delivery policy,
without body-origin rules or exact-call approvals. Both retain 5/5 clean
completions per arm. Fixed attacks induce 15/20 unauthorized raw sends from
unguarded Qwen and 1/20 from Hermes; either guard reduces these to 0/20.
Completion nevertheless falls to 5/20 for Qwen and 10/20 for Hermes in every
arm. In the adaptive study, Qwen has 0/5 successful chains per arm; Hermes has
1/5 unguarded and 0/5 guarded, with three attempts per chain. Agenra denies six
Hermes sends but completes 0/15 attempts. Adaptive inputs differ across arms,
so these rates do not isolate the guard's effect.

No live attack exploited the missing body-origin rule. Replaying the captured
calls admits five of them under the original policy and none under the direct
guard; we exclude these from ASR. The added origin rule rejects all five and
keeps the six legitimate delivery orders.

We also ran three Claude configurations, each on 45 banking and 81 GitHub
episodes. Banking tasks always complete. GitHub completion is 27/27 per arm for medium-effort
Sonnet 5.5; low effort yields 26/27 unguarded and 27/27 guarded; Haiku 5.5 with
thinking disabled yields 18/27 per arm. The only contract violation, under
low-effort Sonnet, is a broader mailbox read before any attack appears. Native
ASR is 0/18 in every arm and configuration, but Haiku reaches only 12/18 attacks,
so these GitHub trials show no security benefit from a guard.
Appendix~\ref{app:evaluation-details} records the protocols and limitations.

\section{Related Work}
\label{sec:related}
\paragraph{Tracking data in agent workflows.}
Information-flow control (IFC) tracks data reaching protected destinations.
CaMeL uses value capabilities and separates data from control flow~\cite{camel};
our matched adapter adds recipient binding, quotas and exact approvals.
FIDES uses recipient-sensitive reader labels and proves planner-loop integrity
noninterference and explicit secrecy~\cite{fides}.
Agenra studies preservation under revisions of ordered task permissions, together
with resource bindings and tracked arguments; it does not establish those planner guarantees.
Prudentia combines policy-aware planning and human interaction with
IFC~\cite{prudentia}; our permits record one-use consent.

\paragraph{Rules governing tool calls.}
FORGE enforces Datalog policies over causal event graphs, including approvals
and workflow preconditions; its guarantee assumes complete mediation and
faithful observation~\cite{pcas}. Agenra studies ordered consumption of task
authority and shares Lean transitions between enforcement and bounded search.
Progent classifies tool-argument policy changes as narrowing
or expansion~\cite{progent}; our revision checks also follow resource bindings
and quota use across calls. AgentSpec and Invariant check runtime rules~\cite{agentspec,invariant2024},
Conseca derives purpose-specific policies~\cite{conseca}, and DRIFT combines
a planned trajectory with dynamic validation~\cite{drift}.
Agenra instead relies on trusted declarations and does not infer intent or
validate generated plans.

\paragraph{Classification, release and changing permissions.}
The label order and trusted-release assumptions follow IFC and decentralized
labels with declassification~\cite{denning76,myersliskov,declassification}.
An exact-call permit approves an effect without changing data classification.
RTBAS, AgentArmor and permissive IFC use model-assisted screening, reconstructed
dependencies or influence-based labels~\cite{rtbas,agentarmor,permissiveifc};
Agenra requires complete dependency declarations for registered computations.
Cedar and Zelkova address request authorization and analysis~\cite{cedar,zelkova},
while UCON studies authorization with mutable attributes~\cite{ucon}.

\paragraph{Consumable authority and policy revision.}
Linear authorization logics express consumable resources~\cite{linearauth};
Bowers et al. implement use-limited credentials with atomic ratification and
protection against nonproductive consumption~\cite{consumablecredentials}.
Agenra's quotas and permits build on these ideas.
Barth and Mitchell already show that widening a digital-rights license can
lose a previously admitted sequence under OMA's selection rule, and prove an
online monotonicity obstruction~\cite{barthdrm2006}. They study a maximal
well-behaved license class and reassignment through linear negation.
We keep the ordered allocator fixed and check a
relation between old and new registrations, so overlapping grants can stay in
place while independent grants cross (Section~\ref{sec:client-theorem}). The relation
follows two ledgers, typed arguments and changing resolver bindings through
failed calls; its certificate checker derives matching coverage from registered
requests. It is sufficient rather than maximal, and performs no reassignment.
Margrave verifies access policies and computes semantic differences between
revisions~\cite{margrave}, including policies with first-applicable rules.
Dougherty et al. distinguish safety from availability and define contextual
containment preserving accessible states under policy replacement~\cite{dynamicpolicies}.
Their policies filter transitions of a supplied environment. In Agenra,
revising a grant can change which ledger entry a call consumes and hence its
successor state. Our relation tracks these ledgers through subsequent calls;
it specializes contextual refinement rather than introducing that notion.

\paragraph{Monitoring and evaluation.}
Reference monitors and edit automata control transitions~\cite{schneider,editautomata};
bounded model checking and Z3 supply the analysis methods~\cite{cbmc,z3}.
AgentDojo, Ajar and AutoDojo motivate joint security and completion measurements
~\cite{agentdojo,ajar,autodojo}; contextual work distinguishes task alignment,
source authorization and data isolation~\cite{formalizingagentsec}.
Finite workflow queries expose interactions between calls that observed
executions can miss.

\section{Discussion and Conclusion}
\label{sec:discussion}
With consumable authority, a revision must be judged against whole workflows,
because a change that admits every individual request can still starve a later
call. Agenra certifies revisions that preserve every declared client under
first-fit allocation and falls back to bounded search when no certificate
applies. Its conditions are sufficient but not necessary, so some preserving
edits remain inconclusive.

The monitor enforces only what the application declares; it cannot recover an
omitted requirement such as the projection input in Section~\ref{sec:motivation}.
Our revision corpus is authored, so how often real maintenance produces
revisions these checks can certify remains open; mining quota and permission
changes from deployed agent applications is a natural next step. We have also
not studied how easily developers write these policies, or how the results vary
across agent populations.

\label{main:end}
\clearpage
\label{statement:start}
\section*{Data-Availability Statement}
Our artifact is a bundle of source code and evidence containing the
Agenra runtime, policies, Lean proofs, fixtures, tests, replay protocols and
recorded evaluation outputs. Its deterministic analyses and saved-trace replays
run without model calls. Reproducing the live LLM trials requires access to the
pinned model deployments, and replaying archived runs does not make them
independent replications. The bundle omits paper sources and superseded
evaluation data; Appendices~\ref{app:artifact} and~\ref{app:evaluation-details}
summarize its reproduction commands and study protocols. We intend to submit the artifact for Artifact Evaluation.

\section*{Use of Generative AI}
Generative AI tools assisted with drafting and revising parts of the manuscript
text and with writing code, including experiment and replay scripts. The authors
directed this use, reviewed and edited all AI-assisted text and code, and take
full responsibility for the content of this paper and its artifact.

\bibliographystyle{ACM-Reference-Format}
\bibliography{references}


\begin{thebibliography}{36}


\ifx \showCODEN    \undefined \def \showCODEN     #1{\unskip}     \fi
\ifx \showISBNx    \undefined \def \showISBNx     #1{\unskip}     \fi
\ifx \showISBNxiii \undefined \def \showISBNxiii  #1{\unskip}     \fi
\ifx \showISSN     \undefined \def \showISSN      #1{\unskip}     \fi
\ifx \showLCCN     \undefined \def \showLCCN      #1{\unskip}     \fi
\ifx \shownote     \undefined \def \shownote      #1{#1}          \fi
\ifx \showarticletitle \undefined \def \showarticletitle #1{#1}   \fi
\ifx \showURL      \undefined \def \showURL       {\relax}        \fi
\providecommand\bibfield[2]{#2}
\providecommand\bibinfo[2]{#2}
\providecommand\natexlab[1]{#1}
\providecommand\showeprint[2][]{arXiv:#2}

\bibitem[Backes et~al\mbox{.}(2018)]%
        {zelkova}
\bibfield{author}{\bibinfo{person}{John Backes}, \bibinfo{person}{Pauline
  Bolignano}, \bibinfo{person}{Byron Cook}, \bibinfo{person}{Catherine Dodge},
  \bibinfo{person}{Andrew Gacek}, \bibinfo{person}{Kasper Luckow},
  \bibinfo{person}{Neha Rungta}, \bibinfo{person}{Oksana Tkachuk}, {and}
  \bibinfo{person}{Carsten Varming}.} \bibinfo{year}{2018}\natexlab{}.
\newblock \showarticletitle{Semantic-based automated reasoning for AWS access
  policies using SMT}. In \bibinfo{booktitle}{\emph{2018 Formal Methods in
  Computer Aided Design (FMCAD)}}. IEEE, \bibinfo{pages}{1--9}.
\newblock


\bibitem[Balunovi{\'c} et~al\mbox{.}(2024)]%
        {invariant2024}
\bibfield{author}{\bibinfo{person}{Mislav Balunovi{\'c}}, \bibinfo{person}{Luca
  Beurer-Kellner}, \bibinfo{person}{Marc Fischer}, {and}
  \bibinfo{person}{Martin Vechev}.} \bibinfo{year}{2024}\natexlab{}.
\newblock \bibinfo{title}{{AI} Agents with Formal Security Guarantees}.
\newblock
\urldef\tempurl%
\url{https://invariantlabs.ai/theme/research/ai_agents_with_formal_security.pdf}
\showURL{%
\tempurl}
\newblock
\shownote{ICML 2024 Next Generation of AI Safety Workshop}.


\bibitem[Barth and Mitchell(2006)]%
        {barthdrm2006}
\bibfield{author}{\bibinfo{person}{Adam Barth} {and} \bibinfo{person}{John~C.
  Mitchell}.} \bibinfo{year}{2006}\natexlab{}.
\newblock \showarticletitle{Managing Digital Rights using Linear Logic}. In
  \bibinfo{booktitle}{\emph{21st Annual IEEE Symposium on Logic in Computer
  Science (LICS)}}. \bibinfo{publisher}{IEEE}, \bibinfo{pages}{127--136}.
\newblock
\href{https://doi.org/10.1109/LICS.2006.32}{doi:\nolinkurl{10.1109/LICS.2006.32}}


\bibitem[Bowers et~al\mbox{.}(2007)]%
        {consumablecredentials}
\bibfield{author}{\bibinfo{person}{Kevin~D. Bowers}, \bibinfo{person}{Lujo
  Bauer}, \bibinfo{person}{Deepak Garg}, \bibinfo{person}{Frank Pfenning},
  {and} \bibinfo{person}{Michael~K. Reiter}.} \bibinfo{year}{2007}\natexlab{}.
\newblock \showarticletitle{Consumable Credentials in Logic-Based
  Access-Control Systems}. In \bibinfo{booktitle}{\emph{Network and Distributed
  System Security Symposium (NDSS)}}. \bibinfo{publisher}{Internet Society}.
\newblock
\urldef\tempurl%
\url{https://www.cs.cmu.edu/~fp/papers/ndss07.pdf}
\showURL{%
\tempurl}


\bibitem[Costa et~al\mbox{.}(2026)]%
        {fides}
\bibfield{author}{\bibinfo{person}{Manuel Costa}, \bibinfo{person}{Boris
  K{\"o}pf}, \bibinfo{person}{Aashish Kolluri}, \bibinfo{person}{Andrew
  Paverd}, \bibinfo{person}{Mark Russinovich}, \bibinfo{person}{Ahmed Salem},
  \bibinfo{person}{Shruti Tople}, \bibinfo{person}{Lukas Wutschitz}, {and}
  \bibinfo{person}{Santiago Zanella-B{\'e}guelin}.}
  \bibinfo{year}{2026}\natexlab{}.
\newblock \showarticletitle{Securing {AI} Agents with Information-Flow
  Control}. In \bibinfo{booktitle}{\emph{Advances in Neural Information
  Processing Systems (NeurIPS)}}.
\newblock
\urldef\tempurl%
\url{https://arxiv.org/abs/2505.23643v2}
\showURL{%
\tempurl}
\newblock
\shownote{Accepted; preprint arXiv:2505.23643}.


\bibitem[Cutler et~al\mbox{.}(2024)]%
        {cedar}
\bibfield{author}{\bibinfo{person}{Joseph~W. Cutler}, \bibinfo{person}{Craig
  Disselkoen}, \bibinfo{person}{Aaron Eline}, \bibinfo{person}{Shaobo He},
  \bibinfo{person}{Kyle Headley}, \bibinfo{person}{Michael Hicks},
  \bibinfo{person}{Kesha Hietala}, \bibinfo{person}{Eleftherios Ioannidis},
  \bibinfo{person}{John Kastner}, \bibinfo{person}{Anwar Mamat},
  \bibinfo{person}{Darin McAdams}, \bibinfo{person}{Matt McCutchen},
  \bibinfo{person}{Neha Rungta}, \bibinfo{person}{Emina Torlak}, {and}
  \bibinfo{person}{Andrew~M. Wells}.} \bibinfo{year}{2024}\natexlab{}.
\newblock \showarticletitle{Cedar: A new language for expressive, fast, safe,
  and analyzable authorization}.
\newblock \bibinfo{journal}{\emph{Proceedings of the ACM on Programming
  Languages}} \bibinfo{volume}{8}, \bibinfo{number}{OOPSLA1}, Article
  \bibinfo{articleno}{118} (\bibinfo{year}{2024}),
  \bibinfo{numpages}{28}~pages.
\newblock
\href{https://doi.org/10.1145/3649835}{doi:\nolinkurl{10.1145/3649835}}


\bibitem[de~Moura and Bj{\o}rner(2008)]%
        {z3}
\bibfield{author}{\bibinfo{person}{Leonardo de Moura} {and}
  \bibinfo{person}{Nikolaj Bj{\o}rner}.} \bibinfo{year}{2008}\natexlab{}.
\newblock \showarticletitle{Z3: An efficient SMT solver}. In
  \bibinfo{booktitle}{\emph{International conference on Tools and Algorithms
  for the Construction and Analysis of Systems}}. Springer,
  \bibinfo{pages}{337--340}.
\newblock


\bibitem[Debenedetti et~al\mbox{.}(2026)]%
        {camel}
\bibfield{author}{\bibinfo{person}{Edoardo Debenedetti}, \bibinfo{person}{Ilia
  Shumailov}, \bibinfo{person}{Tianqi Fan}, \bibinfo{person}{Jamie Hayes},
  \bibinfo{person}{Nicholas Carlini}, \bibinfo{person}{Daniel Fabian},
  \bibinfo{person}{Christoph Kern}, \bibinfo{person}{Chongyang Shi},
  \bibinfo{person}{Andreas Terzis}, {and} \bibinfo{person}{Florian
  Tram{\`e}r}.} \bibinfo{year}{2026}\natexlab{}.
\newblock \showarticletitle{Defeating prompt injections by design}. In
  \bibinfo{booktitle}{\emph{IEEE Conference on Secure and Trustworthy Machine
  Learning (SaTML)}}.
\newblock
\urldef\tempurl%
\url{https://arxiv.org/abs/2503.18813}
\showURL{%
\tempurl}


\bibitem[Debenedetti et~al\mbox{.}(2024)]%
        {agentdojo}
\bibfield{author}{\bibinfo{person}{Edoardo Debenedetti}, \bibinfo{person}{Jie
  Zhang}, \bibinfo{person}{Mislav Balunovic}, \bibinfo{person}{Luca
  Beurer-Kellner}, \bibinfo{person}{Marc Fischer}, {and}
  \bibinfo{person}{Florian Tram{\`e}r}.} \bibinfo{year}{2024}\natexlab{}.
\newblock \showarticletitle{{AgentDojo}: A dynamic environment to evaluate
  prompt injection attacks and defenses for {LLM} agents}.
\newblock \bibinfo{journal}{\emph{Advances in Neural Information Processing
  Systems}}  \bibinfo{volume}{37} (\bibinfo{year}{2024}),
  \bibinfo{pages}{82895--82920}.
\newblock


\bibitem[Denning(1976)]%
        {denning76}
\bibfield{author}{\bibinfo{person}{Dorothy~E Denning}.}
  \bibinfo{year}{1976}\natexlab{}.
\newblock \showarticletitle{A lattice model of secure information flow}.
\newblock \bibinfo{journal}{\emph{Commun. ACM}} \bibinfo{volume}{19},
  \bibinfo{number}{5} (\bibinfo{year}{1976}), \bibinfo{pages}{236--243}.
\newblock


\bibitem[Dougherty et~al\mbox{.}(2006)]%
        {dynamicpolicies}
\bibfield{author}{\bibinfo{person}{Daniel~J. Dougherty}, \bibinfo{person}{Kathi
  Fisler}, {and} \bibinfo{person}{Shriram Krishnamurthi}.}
  \bibinfo{year}{2006}\natexlab{}.
\newblock \showarticletitle{Specifying and Reasoning About Dynamic
  Access-Control Policies}. In \bibinfo{booktitle}{\emph{Automated Reasoning:
  Third International Joint Conference, IJCAR 2006}}.
  \bibinfo{publisher}{Springer}, \bibinfo{pages}{632--646}.
\newblock
\href{https://doi.org/10.1007/11814771_51}{doi:\nolinkurl{10.1007/11814771_51}}


\bibitem[{Dyad Contributors}(2026)]%
        {dyadquota2026}
\bibfield{author}{\bibinfo{person}{{Dyad Contributors}}.}
  \bibinfo{year}{2026}\natexlab{}.
\newblock \bibinfo{title}{Increase Basic Agent free quota from 5 to 10
  messages}.
\newblock \bibinfo{howpublished}{GitHub pull request, merged April 6, 2026}.
\newblock
\urldef\tempurl%
\url{https://github.com/dyad-sh/dyad/pull/3147}
\showURL{%
\tempurl}
\newblock
\shownote{Pinned merge commit 1e56c111167df68533450edf096685f2bb533736}.


\bibitem[Fisler et~al\mbox{.}(2005)]%
        {margrave}
\bibfield{author}{\bibinfo{person}{Kathi Fisler}, \bibinfo{person}{Shriram
  Krishnamurthi}, \bibinfo{person}{Leo~A Meyerovich}, {and}
  \bibinfo{person}{Michael~Carl Tschantz}.} \bibinfo{year}{2005}\natexlab{}.
\newblock \showarticletitle{Verification and change-impact analysis of
  access-control policies}. In \bibinfo{booktitle}{\emph{Proceedings of the
  27th international conference on Software engineering}}.
  \bibinfo{pages}{196--205}.
\newblock


\bibitem[Garg et~al\mbox{.}(2006)]%
        {linearauth}
\bibfield{author}{\bibinfo{person}{Deepak Garg}, \bibinfo{person}{Lujo Bauer},
  \bibinfo{person}{Kevin~D. Bowers}, \bibinfo{person}{Frank Pfenning}, {and}
  \bibinfo{person}{Michael~K. Reiter}.} \bibinfo{year}{2006}\natexlab{}.
\newblock \showarticletitle{A Linear Logic of Authorization and Knowledge}. In
  \bibinfo{booktitle}{\emph{Computer Security -- ESORICS 2006}}
  \emph{(\bibinfo{series}{Lecture Notes in Computer Science},
  Vol.~\bibinfo{volume}{4189})}. \bibinfo{publisher}{Springer},
  \bibinfo{pages}{297--312}.
\newblock
\urldef\tempurl%
\url{https://users.ece.cmu.edu/~lbauer/papers/2006/esorics2006.pdf}
\showURL{%
\tempurl}


\bibitem[Greshake et~al\mbox{.}(2023)]%
        {greshake23}
\bibfield{author}{\bibinfo{person}{Kai Greshake}, \bibinfo{person}{Sahar
  Abdelnabi}, \bibinfo{person}{Shailesh Mishra}, \bibinfo{person}{Christoph
  Endres}, \bibinfo{person}{Thorsten Holz}, {and} \bibinfo{person}{Mario
  Fritz}.} \bibinfo{year}{2023}\natexlab{}.
\newblock \showarticletitle{Not what you've signed up for: Compromising
  real-world {LLM}-integrated applications with indirect prompt injection}. In
  \bibinfo{booktitle}{\emph{Proceedings of the 16th ACM workshop on artificial
  intelligence and security}}. \bibinfo{pages}{79--90}.
\newblock


\bibitem[Kolluri et~al\mbox{.}(2026)]%
        {prudentia}
\bibfield{author}{\bibinfo{person}{Aashish Kolluri}, \bibinfo{person}{Rishi
  Sharma}, \bibinfo{person}{Manuel Costa}, \bibinfo{person}{Boris K{\"o}pf},
  \bibinfo{person}{Tobias Nie{\ss}en}, \bibinfo{person}{Mark Russinovich},
  \bibinfo{person}{Shruti Tople}, {and} \bibinfo{person}{Santiago
  Zanella-B{\'e}guelin}.} \bibinfo{year}{2026}\natexlab{}.
\newblock \showarticletitle{Optimizing Agent Planning for Security and
  Autonomy}. In \bibinfo{booktitle}{\emph{International Conference on Learning
  Representations (ICLR)}}.
\newblock
\urldef\tempurl%
\url{https://openreview.net/forum?id=g0aVCDY3gS}
\showURL{%
\tempurl}


\bibitem[Kroening et~al\mbox{.}(2023)]%
        {cbmc}
\bibfield{author}{\bibinfo{person}{Daniel Kroening}, \bibinfo{person}{Peter
  Schrammel}, {and} \bibinfo{person}{Michael Tautschnig}.}
  \bibinfo{year}{2023}\natexlab{}.
\newblock \showarticletitle{{CBMC}: The {C} Bounded Model Checker}.
\newblock \bibinfo{journal}{\emph{arXiv preprint arXiv:2302.02384}}
  (\bibinfo{year}{2023}).
\newblock
\urldef\tempurl%
\url{https://arxiv.org/abs/2302.02384}
\showURL{%
\tempurl}


\bibitem[Li et~al\mbox{.}(2025)]%
        {drift}
\bibfield{author}{\bibinfo{person}{Hao Li}, \bibinfo{person}{Xiaogeng Liu},
  \bibinfo{person}{Hung-Chun Chiu}, \bibinfo{person}{Dianqi Li},
  \bibinfo{person}{Ning Zhang}, {and} \bibinfo{person}{Chaowei Xiao}.}
  \bibinfo{year}{2025}\natexlab{}.
\newblock \showarticletitle{{DRIFT}: Dynamic Rule-Based Defense with Injection
  Isolation for Securing {LLM} Agents}. In \bibinfo{booktitle}{\emph{Advances
  in Neural Information Processing Systems}}, Vol.~\bibinfo{volume}{38}.
\newblock
\urldef\tempurl%
\url{https://proceedings.neurips.cc/paper_files/paper/2025/hash/77f3b26c7907aa27b207df9b9d43f29a-Abstract-Conference.html}
\showURL{%
\tempurl}


\bibitem[Ligatti et~al\mbox{.}(2005)]%
        {editautomata}
\bibfield{author}{\bibinfo{person}{Jay Ligatti}, \bibinfo{person}{Lujo Bauer},
  {and} \bibinfo{person}{David Walker}.} \bibinfo{year}{2005}\natexlab{}.
\newblock \showarticletitle{Edit automata: Enforcement mechanisms for run-time
  security policies}.
\newblock \bibinfo{journal}{\emph{International Journal of Information
  Security}} \bibinfo{volume}{4}, \bibinfo{number}{1--2}
  (\bibinfo{year}{2005}), \bibinfo{pages}{2--16}.
\newblock


\bibitem[Ma et~al\mbox{.}(2026)]%
        {autodojo}
\bibfield{author}{\bibinfo{person}{Xinhang Ma}, \bibinfo{person}{Taoran Li},
  \bibinfo{person}{Chaowei Xiao}, \bibinfo{person}{Zhiyuan Yu},
  \bibinfo{person}{Ning Zhang}, {and} \bibinfo{person}{Yevgeniy Vorobeychik}.}
  \bibinfo{year}{2026}\natexlab{}.
\newblock \showarticletitle{{AutoDojo}: A Generative Benchmark for Evaluating
  Prompt Injection Defenses in {LLM} Agents}.
\newblock \bibinfo{journal}{\emph{arXiv preprint arXiv:2606.15057}}
  (\bibinfo{year}{2026}).
\newblock
\urldef\tempurl%
\url{https://arxiv.org/abs/2606.15057v3}
\showURL{%
\tempurl}
\newblock
\shownote{Version 3, revised September 28, 2026}.


\bibitem[Myers and Liskov(2000)]%
        {myersliskov}
\bibfield{author}{\bibinfo{person}{Andrew~C Myers} {and}
  \bibinfo{person}{Barbara Liskov}.} \bibinfo{year}{2000}\natexlab{}.
\newblock \showarticletitle{Protecting privacy using the decentralized label
  model}.
\newblock \bibinfo{journal}{\emph{ACM Transactions on Software Engineering and
  Methodology (TOSEM)}} \bibinfo{volume}{9}, \bibinfo{number}{4}
  (\bibinfo{year}{2000}), \bibinfo{pages}{410--442}.
\newblock


\bibitem[{Open Mobile Alliance}(2008)]%
        {omarel2008}
\bibfield{author}{\bibinfo{person}{{Open Mobile Alliance}}.}
  \bibinfo{year}{2008}\natexlab{}.
\newblock \bibinfo{booktitle}{\emph{{DRM} Rights Expression Language, Approved
  Version 2.0.1}}.
\newblock
\urldef\tempurl%
\url{https://www.openmobilealliance.org/release/DRM/V2_0_1-20080226-A/OMA-TS-DRM_REL-V2_0_1-20080226-A.pdf}
\showURL{%
\tempurl}
\newblock
\shownote{OMA-TS-DRM\_REL-V2\_0\_1-20080226-A, Section 5.9}.


\bibitem[Palumbo et~al\mbox{.}(2026)]%
        {pcas}
\bibfield{author}{\bibinfo{person}{Nils Palumbo}, \bibinfo{person}{Sarthak
  Choudhary}, \bibinfo{person}{Jihye Choi}, \bibinfo{person}{Guy Amir},
  \bibinfo{person}{Prasad Chalasani}, \bibinfo{person}{Isil Dillig}, {and}
  \bibinfo{person}{Somesh Jha}.} \bibinfo{year}{2026}\natexlab{}.
\newblock \showarticletitle{Formal Policy Enforcement for Agentic Systems}.
\newblock \bibinfo{journal}{\emph{arXiv preprint arXiv:2602.16708v4}}
  (\bibinfo{year}{2026}).
\newblock
\urldef\tempurl%
\url{https://arxiv.org/abs/2602.16708v4}
\showURL{%
\tempurl}
\newblock
\shownote{Version 4, revised October 12, 2026}.


\bibitem[Park and Sandhu(2004)]%
        {ucon}
\bibfield{author}{\bibinfo{person}{Jaehong Park} {and} \bibinfo{person}{Ravi
  Sandhu}.} \bibinfo{year}{2004}\natexlab{}.
\newblock \showarticletitle{The UCONABC Usage Control Model}.
\newblock \bibinfo{journal}{\emph{ACM Transactions on Information and System
  Security}} \bibinfo{volume}{7}, \bibinfo{number}{1} (\bibinfo{year}{2004}),
  \bibinfo{pages}{128--174}.
\newblock
\href{https://doi.org/10.1145/984334.984339}{doi:\nolinkurl{10.1145/984334.984339}}


\bibitem[{Revenera}(2026a)]%
        {flexnetallocation2026}
\bibfield{author}{\bibinfo{person}{{Revenera}}.}
  \bibinfo{year}{2026}\natexlab{a}.
\newblock \bibinfo{booktitle}{\emph{{FlexNet Embedded} 2026.04 License Server
  Administration Guide}}.
\newblock
\urldef\tempurl%
\url{https://docs.revenera.com/fne/2026_04/pdf/FNE_LicenseServerAdminGuide_2026-04.pdf}
\showURL{%
\tempurl}
\newblock
\shownote{Server Behavior when Assigning Features to Clients, p. 162}.


\bibitem[{Revenera}(2026b)]%
        {flexnetaccess2026}
\bibfield{author}{\bibinfo{person}{{Revenera}}.}
  \bibinfo{year}{2026}\natexlab{b}.
\newblock \bibinfo{booktitle}{\emph{Use Case: Using Regular Expression to
  Allocate License Counts}}.
\newblock
\urldef\tempurl%
\url{https://docs.revenera.com/fne/2026_06/adminguide/Content/helplibrary/Use_Case__Using_Regular_Expression_to_Allocate_License_Counts.htm}
\showURL{%
\tempurl}
\newblock
\shownote{FlexNet Embedded 2026.06 License Server Administration Guide}.


\bibitem[Sabelfeld and Sands(2009)]%
        {declassification}
\bibfield{author}{\bibinfo{person}{Andrei Sabelfeld} {and}
  \bibinfo{person}{David Sands}.} \bibinfo{year}{2009}\natexlab{}.
\newblock \showarticletitle{Declassification: Dimensions and principles}.
\newblock \bibinfo{journal}{\emph{Journal of Computer Security}}
  \bibinfo{volume}{17}, \bibinfo{number}{5} (\bibinfo{year}{2009}),
  \bibinfo{pages}{517--548}.
\newblock


\bibitem[Schneider(2000)]%
        {schneider}
\bibfield{author}{\bibinfo{person}{Fred~B Schneider}.}
  \bibinfo{year}{2000}\natexlab{}.
\newblock \showarticletitle{Enforceable security policies}.
\newblock \bibinfo{journal}{\emph{ACM Transactions on Information and System
  Security (TISSEC)}} \bibinfo{volume}{3}, \bibinfo{number}{1}
  (\bibinfo{year}{2000}), \bibinfo{pages}{30--50}.
\newblock


\bibitem[Sharma et~al\mbox{.}(2026)]%
        {ajar}
\bibfield{author}{\bibinfo{person}{Reshabh~K Sharma}, \bibinfo{person}{Linxi
  Jiang}, \bibinfo{person}{Shuo Chen}, {and} \bibinfo{person}{Zhiqiang Lin}.}
  \bibinfo{year}{2026}\natexlab{}.
\newblock \showarticletitle{Ajar: Measuring Open Privilege in Agent Defenses}.
\newblock \bibinfo{journal}{\emph{arXiv preprint arXiv:2609.26900}}
  (\bibinfo{year}{2026}).
\newblock
\urldef\tempurl%
\url{https://arxiv.org/abs/2609.26900}
\showURL{%
\tempurl}


\bibitem[Shi et~al\mbox{.}(2026)]%
        {progent}
\bibfield{author}{\bibinfo{person}{Tianneng Shi}, \bibinfo{person}{Jingxuan
  He}, \bibinfo{person}{Zhun Wang}, \bibinfo{person}{Hongwei Li},
  \bibinfo{person}{Linyu Wu}, \bibinfo{person}{Wenbo Guo}, {and}
  \bibinfo{person}{Dawn Song}.} \bibinfo{year}{2026}\natexlab{}.
\newblock \showarticletitle{{Progent}: Securing {AI} Agents with Privilege
  Control}.
\newblock \bibinfo{journal}{\emph{arXiv preprint arXiv:2504.11703}}
  (\bibinfo{year}{2026}).
\newblock
\urldef\tempurl%
\url{https://arxiv.org/abs/2504.11703v3}
\showURL{%
\tempurl}
\newblock
\shownote{Version 3, revised May 14, 2026; first released in 2025}.


\bibitem[Siddiqui et~al\mbox{.}(2025)]%
        {permissiveifc}
\bibfield{author}{\bibinfo{person}{Shoaib~Ahmed Siddiqui},
  \bibinfo{person}{Radhika Gaonkar}, \bibinfo{person}{Boris K{\"o}pf},
  \bibinfo{person}{David Krueger}, \bibinfo{person}{Andrew Paverd},
  \bibinfo{person}{Ahmed Salem}, \bibinfo{person}{Shruti Tople},
  \bibinfo{person}{Lukas Wutschitz}, \bibinfo{person}{Menglin Xia}, {and}
  \bibinfo{person}{Santiago Zanella-B{\'e}guelin}.}
  \bibinfo{year}{2025}\natexlab{}.
\newblock \showarticletitle{Permissive Information-Flow Analysis for Large
  Language Models}.
\newblock \bibinfo{journal}{\emph{Transactions on Machine Learning Research}}
  (\bibinfo{year}{2025}).
\newblock
\urldef\tempurl%
\url{https://openreview.net/forum?id=ufYRO8y3mr}
\showURL{%
\tempurl}


\bibitem[Siu et~al\mbox{.}(2026)]%
        {formalizingagentsec}
\bibfield{author}{\bibinfo{person}{Vincent Siu}, \bibinfo{person}{Jingxuan He},
  \bibinfo{person}{Kyle Montgomery}, \bibinfo{person}{Zhun Wang},
  \bibinfo{person}{Neil Gong}, \bibinfo{person}{Chenguang Wang}, {and}
  \bibinfo{person}{Dawn Song}.} \bibinfo{year}{2026}\natexlab{}.
\newblock \showarticletitle{A framework for formalizing LLM agent security}.
\newblock \bibinfo{journal}{\emph{arXiv preprint arXiv:2603.19469}}
  (\bibinfo{year}{2026}).
\newblock


\bibitem[Tsai and Bagdasarian(2025)]%
        {conseca}
\bibfield{author}{\bibinfo{person}{Lillian Tsai} {and} \bibinfo{person}{Eugene
  Bagdasarian}.} \bibinfo{year}{2025}\natexlab{}.
\newblock \showarticletitle{Contextual agent security: A policy for every
  purpose}. In \bibinfo{booktitle}{\emph{Proceedings of the 2025 Workshop on
  Hot Topics in Operating Systems}}. \bibinfo{pages}{8--17}.
\newblock


\bibitem[Wang et~al\mbox{.}(2026)]%
        {agentspec}
\bibfield{author}{\bibinfo{person}{Haoyu Wang}, \bibinfo{person}{Christopher~M.
  Poskitt}, {and} \bibinfo{person}{Jun Sun}.} \bibinfo{year}{2026}\natexlab{}.
\newblock \showarticletitle{{AgentSpec}: Customizable Runtime Enforcement for
  Safe and Reliable {LLM} Agents}. In \bibinfo{booktitle}{\emph{Proceedings of
  the 48th IEEE/ACM International Conference on Software Engineering (ICSE)}}.
  \bibinfo{publisher}{ACM}, \bibinfo{pages}{2938--2950}.
\newblock
\href{https://doi.org/10.1145/3744916.3764546}{doi:\nolinkurl{10.1145/3744916.3764546}}


\bibitem[Wang et~al\mbox{.}(2025)]%
        {agentarmor}
\bibfield{author}{\bibinfo{person}{Peiran Wang}, \bibinfo{person}{Yang Liu},
  \bibinfo{person}{Yunfei Lu}, \bibinfo{person}{Yifeng Cai},
  \bibinfo{person}{Hongbo Chen}, \bibinfo{person}{Qingyou Yang},
  \bibinfo{person}{Jie Zhang}, \bibinfo{person}{Jue Hong}, {and}
  \bibinfo{person}{Ye Wu}.} \bibinfo{year}{2025}\natexlab{}.
\newblock \showarticletitle{Securing Large Language Model Agents via Structured
  Graph Abstraction}.
\newblock \bibinfo{journal}{\emph{arXiv preprint arXiv:2508.01249v3}}
  (\bibinfo{year}{2025}).
\newblock
\urldef\tempurl%
\url{https://arxiv.org/abs/2508.01249v3}
\showURL{%
\tempurl}


\bibitem[Zhong et~al\mbox{.}(2025)]%
        {rtbas}
\bibfield{author}{\bibinfo{person}{Peter~Yong Zhong}, \bibinfo{person}{Siyuan
  Chen}, \bibinfo{person}{Ruiqi Wang}, \bibinfo{person}{McKenna McCall},
  \bibinfo{person}{Ben~L Titzer}, \bibinfo{person}{Heather Miller}, {and}
  \bibinfo{person}{Phillip~B Gibbons}.} \bibinfo{year}{2025}\natexlab{}.
\newblock \showarticletitle{{RTBAS}: Defending {LLM} agents against prompt
  injection and privacy leakage}.
\newblock \bibinfo{journal}{\emph{arXiv preprint arXiv:2502.08966}}
  (\bibinfo{year}{2025}).
\newblock


\end{thebibliography}
\clearpage
\appendix
\section{Supporting Formal Arguments}
\label{appendix:start}
\label{app:formal-coverage}\label{app:proofs}
\label{app:labels}\label{app:transfers}\label{app:path-decisions}
The arguments retain the main text's hypotheses. Admission concerns checks at
the dispatch point; later lookup conflicts do not change past decisions.

\begin{lemma}[Transfer evaluation]\label{lem:transfer-evaluation}
For fixed valid parameters and labels, inheritance preserves the input; a transform
removes matching categories, subtracts sensitivity with saturation at Public,
then applies explicit sensitivity and trust assignments, in that order.
The Python reference function \texttt{EdgeDef.compute\_output\_label} implements
these updates.
\end{lemma}
\begin{proof}
Inspect the corresponding component updates. Category matching includes a path
and its descendants; label ordering instead uses exact set inclusion. The trust
order is $true\sqsubseteq false$, making conjunction the join. Constant or
selective-opaque transfers can discard distinctions and require trusted
configuration; scans are outside the closed-core transfer.
\end{proof}

\begin{theorem}[Enforcement at admission]\label{thm:enforcement}
Under the closed-core hypotheses, each reachable committed entry passed the
structural, flow, readiness, value, path and contract guards in its admission state. Both ledgers cover each effect,
and finite per-permission counters remain within quota.
\end{theorem}
\begin{proof}
Only Allow appends an entry. First-fit increments a temporary counter only below
capacity, and failed preparation commits neither side. Fresh IDs and earlier
parents preserve map domains and acyclicity. Designated successful observations
bind an unbound slot, preserve an equal binding or invalidate a malformed/conflicting
one; invalidity is absorbing. Later observations do not change past admission.
\end{proof}

\begin{lemma}[Value-approval recurrence]\label{lem:value-approval}
An approval edge marks its output approved. A preserving edge retains approval
only with one distinct approved parent. Other outputs are unapproved. At a
protected edge, acceptance requires nonempty parents, all approved.
\end{lemma}
\begin{proof}
Induct over the trace, deduplicating parents; missing entries are false.
This review rule does not consume an exact-call permit.
\end{proof}

\begin{lemma}[Plain-edge path guard]\label{prop:path}
For $e_1\to^*e_2$, lineage/session scope and no condition callback, let
$\mathcal M$ be the matching intervals, including endpoints, and $K$ the
alternative-edge set. Forbid accepts iff each interval contains an edge from
nonempty $K$. Require with nonempty $K$ accepts off sinks, with no trigger, or
when some interval contains an edge from $K$.
With empty $K$, ordinary require at a sink needs a matching interval.
\end{lemma}
\begin{proof}
The reference evaluator enumerates ordered endpoint pairs and uses the stated
universal/existential tests. A full-path API instead returns violation intervals,
so its require guard accepts when none remain. It is not an ordinary positive
matcher. These rules and general sequence/repetition matching are outside
$\mathcal F$; deployed Lean matcher correspondence is distinguished below.
\end{proof}

\paragraph{Correspondence and adapter inclusion.}
\label{app:correspondence}\label{app:adapter-projection}
Under $S_i\sim\widehat S_i$, joins preserve maximum sensitivity, category union
and trust conjunction; ordered transforms and finite predicate tables therefore
agree. Typed argument checks and effect expansion are finite constants. An
injective positive resource ID represents each possible resolver result; zero
means unbound and a distinct sentinel makes undeclared candidate resources
nonmatching. The least eligible permission is the same on both sides, and only
Allow commits its reservations. Successful lookup cases agree as above; failed
calls observe nothing. Induction establishes both directions of
Theorem~\ref{thm:correspondence}; for fixed choices every related decision,
counter and binding is determined. A denied intermediate step contradicts the
prefix constraint; a denied final candidate remains representable.

The restricted adapter uses immutable whole handles and literals, serial calls,
and no permits, scans, reset or serialization/observation failures. Each handle
has an earlier successful producer. Exhaustive role assignments retain aliases
and ordered deduplicated parents; outcome coverage retains its real result or
terminal exception. State-preserving denials may be removed only when Bad does
not inspect their history. Every real assignment witnesses $C_{\rm value}$,
and outcome coverage witnesses $C_{\rm outcome}$; hence their conjunction cannot
remove a covered execution. It can still admit unrealizable combinations, so
SAT requires replay. These steps discharge the inclusion used by refined
bounded exclusion, not real-service coverage itself.

\paragraph{Exact-call composition.}\label{app:permit-symbolic}
Initially no permit is issued. Each protected dispatch requires a previously
issued, matching unused identity and consumes it before invocation. Fresh Issue,
Return and Failure cannot undo consumption; serialization excludes concurrent
reuse. Immutable handles and literals connect the signature to actual arguments.
For finite scheduled offers, equate issued/consumed bits and producer identities
with their concrete counterparts. Optional issuance gives the same permit guard;
only base Allow commits both core state and consumption. Induction gives
correspondence, and domain/outcome/offer coverage gives executor inclusion.
Resets, arbitrary approval arrival and remote exactly-once effects are excluded.

\section{Lean Coverage and Artifact}
\label{app:lean-core}\label{app:artifact}\label{app:implementation-details}
\paragraph{Connecting proofs to the tool.}
The Lean 4.19 development's audit rejects proof holes and added project axioms,
allowing only standard foundational dependencies. Some proofs concern deployed
checking functions; others concern companion models.
\texttt{RevisionCertificate} validates the swap/extension relation proved by
\texttt{ClientPreservation}. \texttt{FiniteLedgerSimulation} validates closed
counter-state relations on finite-quota ledgers; its closure theorem covers all
finite words in the declared effect domain. \texttt{RuntimeRevision} transports
either relation through contract calls. \texttt{RegisteredCoverage} covers actual tool expansion for
arbitrary resolver states, and \texttt{RequestFamilyCoverage} generalizes
uninspected payloads to arbitrary JSON. \texttt{EngineRevision} preserves full
safety state through serialized calls, scans, failures and reviewer updates,
deriving equal nonledger decisions. Identical initial safety states, fixed configuration, identical external observations,
request coverage and trusted translation/dispatch remain premises; interleaved
pending operations, resets and post-denial states are excluded.
\texttt{RuntimeCausal} derives a deployed-state invariant and connects supported
runtime transitions to \texttt{CausalSafety}'s tree theorem. Dependencies, source
labels and release behavior remain trusted. Generic contract correspondence abstracts noncontract checks
as \texttt{Call.ready}, rather than mechanizing all of $\mathcal F$.

The deployed \texttt{RuntimeCore}/\texttt{Safety} functions use proved allocation,
resolver observation, ordered lookup and one-use permit operations. Selected
theorems cover label joins, commit guards, consumption before dispatch and atomic
safety/quota commit. Exact structural signatures bind edges, ordered handles and
typed arguments; opaque Python object-to-ID mapping and authenticated issuance
remain trusted. \texttt{PolicySemantics} proves flow admission;
\texttt{HistorySemantics} proves the deployed matcher sound and complete against
independent interval relations, including nested sequences, choice, repetitions
and negation. Repetitions advance; nesting depth is at most 256. Decoding, primitive label
tests and Python lowering remain boundaries. \texttt{Provenance} connects the deployed
argument-origin test to an inductive producer/preserver relation;
\texttt{DeliveryApplication} connects the delivery query to origin, recipient,
call-count and permit requirements.
\texttt{BoundedSearch} proves sound and complete search against a separate
reachability relation. \texttt{StreamingSearch} aggregates the same tree exactly
without retaining it. The \texttt{Workflow} theorems
\path{checked_excludes} and \path{witness_reachable} connect query results to it.
Scanner observations and external outcome/reviewer coverage are supplied premises.

Certificates check selected emitted Z3 equations for all valuations.
Separate certificates check complete formulas for selected finite instances,
using trusted syntax export and fresh runtime replay tables.
These checks do not prove the Python compiler correct for every input or
certify arbitrary Z3 UNSAT answers. Optional state-indexed outcome
restrictions likewise require coverage of their supplied service transitions.
Tool declarations, JSON transport, compiler/runtime, scanner judgments and hidden
effects remain boundaries. Framework callbacks must propagate rejection; logging
a callback exception and dispatching anyway violates mediation. Resets must be
serialized, and old completions must not update a new task.

\paragraph{Reproduction.}
The artifact's quick runner validates exact required test identities, source
digests and semantic evidence; the separate Lean gate checks the proofs and
translation certificates. Current application, revision and scaling studies are offline, separate from
archived planner trials. From the root:
\begin{quote}\small
\begin{verbatim}
(cd formal/lean && lake build)
PYTHON_BIN="$PWD/.venv/bin/python" bash formal/lean/check.sh
.venv/bin/python -m experiments.lean_application_evaluation
.venv/bin/python -m experiments.declared_origin_delivery
.venv/bin/python artifact/oopsla/reproduce.py quick \
  --output /tmp/agenra-results
\end{verbatim}
\end{quote}

\section{Evaluation Coverage and Provenance}
\label{app:evaluation-details}\label{app:allocation-alternatives}
\paragraph{Which workflows and services were checked.}\label{par:client-preservation}
Revision controls use A, B, AB and BA batches, split AB/BA calls, and forbidden
AA, BB and C. Lean's \texttt{AllocationAlternatives} checks the feasible batch
repair and a split counterexample: local allocation can consume B's later grant.
Existential completion is weaker than preservation under every allocation.
Branching deletion checks 36 policy/manifest/environment contexts through four
calls, including mixed errors and transport failure. Banking and GitHub use
pinned local services with authored domains, registrations and additional
oracles. Outcome coverage must include installed faults; fixed schedules and
storage replay do not establish arbitrary coverage. Runtime, value-consistent
and outcome-consistent witnesses remain distinct.

\paragraph{Keeping experiment protocols separate.}
Unguarded, authored direct-guard and Agenra arms share services. Earlier v2/v3
studies retain separate archives. Typed v3 supplies steps and roles.
The new delivery protocol schedules 120 episodes per model:
clean plus four fixed tactics, five repetitions and three arms (75), then five
three-attempt adaptive chains per arm (45). Only fixed cases match payloads and
seeds across arms; adaptive payloads depend on history. Historical qualification
and fresh neutral competence gates are separate. Trials use the original policy
with its body-origin omission. Effect-free lineage captures are excluded from
ASR. Both complete the schedule. Hermes first loses transport after nine episodes;
its error and 110 unrun cases are archived separately. A repaired SSH route
allows a full rerun with identical protocol and sources, without pooling outcomes.
Transcripts do not attest remote weights or HTTP bytes.

Claude uses native Messages: Sonnet 5.5 with adaptive thinking at medium/low
effort, and Haiku 5.5 with thinking disabled at medium effort. Each has three
repetitions per condition/arm, 45 banking episodes (three tools) and 81 GitHub
episodes (three tasks, 34 tools). All 378 episodes finish without API errors;
workflow failures remain in the denominators. Native histories are audited and
local service effects replay exactly. Banking's attacks have no native ASR oracle.
The historical medium baseline is not interleaved with low effort; Haiku changes
both model and thinking mode. These selected fixed attacks do not exercise
declared-origin delivery or establish population robustness.

\label{appendix:end}
\end{document}